%% file: main.tex
\documentclass[a4paper,11pt]{article}
\usepackage{jheppub} 
\usepackage{lineno}
\usepackage{comment}
\usepackage{xcolor}
\usepackage{xfrac}
\usepackage{amsmath}
\usepackage{amsfonts}
\usepackage{amsthm}
\usepackage{thmtools}
\usepackage{thm-restate}
\usepackage{amssymb}
\usepackage{slashed}
\usepackage{accents}
\usepackage{proof-at-the-end}
\usepackage{tensor}
\usepackage{graphicx}

\usepackage{tikz}
\usepackage{tikz-feynman}
\tikzfeynmanset{warn luatex=false}
\usetikzlibrary{shapes.geometric, arrows, patterns, decorations.markings, arrows.meta}
\usetikzlibrary{calc}
\tikzstyle{top} = [rectangle, rounded corners, minimum width=3cm, minimum height=1cm, text centered, text width=3cm, draw=black, fill=red!30]
\allowdisplaybreaks

\input{pics}

\definecolor{dbrown}{HTML}{C70039}

\def\bnon{\begin{equation*}}
\def\enon{\end{equation*}}
\def\bnu{\begin{equation}}
\def\enu{\end{equation}}
\def\bqa{\begin{eqnarray*}}
\def\eqa{\end{eqnarray*}}
\def\blsd{\begin{enumerate}}
\def\elsd{\end{enumerate}}
\def\blst{\begin{itemize}}
\def\elst{\end{itemize}}
\def\eps{\varepsilon}
\def\ide{\mathbb{I}}

\def\lgen{\mathcal{L}^{\text{gen}}}
\def\lgop{\mathcal{L}^{\text{gen}}_{\text{op},N}}

\def\dij{\delta^{IJ}}
\def\zab{Z_{AB}^I}

\def\tr{\text{Tr}}

\declaretheorem[name=Theorem]{theo}
\declaretheorem[sibling=theo,name=Corollary]{cor}

\preprint{PUBDB-2026-00812}

\title{\boldmath Connecting Supersymmetry to Non-Supersymmetric theories: the Gross-Neveu-Yukawa example}

\author[a]{Mrigankamauli Chakraborty}
\emailAdd{mrigankamauli.chakraborty@desy.de}
\affiliation[a]{II. Institute for Theoretical Physics, Hamburg University\\
Luruper Chaussee 149, D-22761 Hamburg, Germany}

\author[a]{and Sven-Olaf Moch}
\emailAdd{sven-olaf.moch@desy.de}
\affiliation

\begin{abstract}{
We construct a generalized Lagrangian that unifies the Gross–Neveu–Yukawa, Nambu–Jona-Lasinio–Yukawa, and Wess–Zumino models, allowing for arbitrary scalar and fermion flavors in $D$-dimensional regularization. 
This framework clarifies how emergent supersymmetry arises at critical points and reveals structural connections between these theories. 
The unified formulation provides additional supersymmetry Ward identities that simplify loop calculations, even for non-supersymmetric models. 
As an application, we show how this technique can reduce the computational cost of determining anomalous dimensions of twist‑two operators.
}\end{abstract}

\begin{document}

\maketitle
\if{1==0}
\newpage

{\color{blue}
 Hi Sven. Thanks for the draft paper which I've read through. Apologies
for the delay but it arrived when I had a class test to mark and administer
at the start of the week. I've enjoyed reading through the analysis. One
thing is that it touchs on some of my thinking that I arrived at recently in
relation to the large N expansion of the chiral Gross-Neveu which is the
fermion sector of the Wess-Zumino model. Part of that will be briefly
included in my Bayreuth talk.

Concerning the general Lagrangian set up the component Wess-Zumino
Lagrangian has complex scalars $\phi$ and $\bar{\phi}$ which derived from the
two chiral superfields. My reference point on this is the Avdeev, Gorishny,
Kamenshchik and Larin paper Phys Lett B117 (1982), 321.
So in figures 2
and 3 the graphs with arrows would have arrows on all the propagators.

{\color{teal}$\implies$ As the figures 2 and 3 come from the generalised Lagrangian that has real scalars, we do not have the arrows on the scalar propagators, whereas the Avdeev et al paper uses complex scalar fields and hence has directed propagators. If they also chose two real scalars instead of one complex, they would also have no arrows.}

{\color{red}
Added the refererence and footnote~\ref{ft:scalars} to the text and captions of figure 2 and 3.
}

This doesn't invalidate your reasoning. If two lines at a vertex have arrows
directed to the interaction point that is sufficient for the reduction in the
graphs contributing to a Green's function. In other words the only graphs
that can arise are those with even edge subgraphs.

On the $\gamma^5$ discussion on page 15 the right graph of figure 3 is a
vertex graph. So strictly one would not need to evaluate it in the
Wess-Zumino model. Instead it would be the 2-point counterpart given
by joining the three open lines at a point and then cutting a boson line
to give a six loop graph which may be nonplanar. So $\gamma^5$ may not
be a problem before six loops unless there is another topology that is a
counterexample to this.

{\color{teal}$\implies$ This is a very excellent point, we have to look into this. Priority 2.}

Related to this when I was looking at a different
paper I came across an article by Kamenshchik which was produced
after their four loop Wess-Zumino paper where it is claimed that the
supersymmetric Ward identity could be violated starting from five loops.
The reference is Theoretical and Mathematical Physics 55 (1983) 431
and I think it does some of the analysis in the component Lagrangian. It's
not easy to decipher what has been derived there but it looks relevant for
that part of your paper.

{\color{red}
Added the refererence and footnote~5 to the text.
}

In section 4.4 where you determine the conditions for the two anomalous
dimensions to be equal there is an additional condition that has to be
satisfied for emergent supersymmetry. The other property is that both
beta-functions should coincide for the two sets of $n_s$ and $n_f$ values. While
this is a moot point it should follow from the first condition on page 20. What
I'm worried about is whether you have completely established the emergent
supersymmetry. For instance do the Wess-Zumino renormalisation group
equations appear under one of the two emergent conditions? This may not
be possible because of scheme differences beyond two loops. However I
couldn't see it for one and two loops from equations (4.6) and (4.7). On a
separate but related point it would be useful to have the four loop anomalous
dimensions recorded somewhere and available in electronic form for others
to use. The other item for ensuring the emergent supersymmetry is credible
is that the non-renormalisation of the coupling constant appears. In other
words the emergent beta-function is proportional to the field anomalous
dimensions. I couldn't see any comment about these points.

{\color{teal}$\implies$ Priority 1.}
{\color{black}
\begin{itemize}
    \item We did not calculate the 4 loop quartic beta for IR rearrangement problems, but we can show that the 4 loop yukawa beta is exactly equal to the ones calculated in the zerf papers. {\color{teal} Added.}
    \item non-renormalisation checked. {\color{teal} Added.}
    \item how far exactly to go into the RGE. {\color{teal} What about this?}
    {\color{red} It is OK as is.}
\end{itemize}}

At the end of section 4.4 another point you could make is that using the
generalized Lagrangian you have in fact shown there are no other emergent
supersymmetric theories.

The relation of (4.14) can be seen I think with an adaptation of the cut and
glue technique of 1004.1153
If the open ends of each 2-point graph are
closed with a propagator then the same topology occurs. The factor of 2
may come from a symmetry factor.

{\color{red}
Added the refererence and footnote~6 to the text.
}

I had a read through the operator formalism of section but haven't any
immediate comments. It looks as if things are going in the right direction
but of course the big test is that of dealing with supersymmetric gauge
theories. The Wess-Zumino emergent supersymmetry case may have
some hidden subtlety that allows the operator analysis to proceed. 

I hope these comments make sense but I may not have got the intricate
detail that I'm trying to convey phrased correctly or clearly. Apologies if
so and anything is unclear. We certainly need to have a long discussion
in Bayreuth, which I'm looking forward to, as there are a lot of things that
could be done and followed up on in this area.

Best, John.
}
\newpage
\fi

\section{Introduction}

Supersymmetry (SUSY) is a theoretical symmetry between bosons and fermions, linking particles with their superpartners. 
Originally developed in high-energy physics, SUSY aims to unify matter and forces, see, e.g.~\cite{Martin:1997ns}.
Despite lacking experimental confirmation, its mathematical elegance and potential to solve fundamental problems continue to inspire extensive research~\cite{ParticleDataGroup:2024cfk}. 
SUSY can also emerge as an effective symmetry in quantum field theory and condensed matter physics, particularly at low energies or critical points where systems exhibit enhanced symmetries. 
This emergent SUSY provides a valuable bridge between high-energy concepts and condensed matter phenomena, offering insights into strongly correlated systems and quantum phase transitions, and often simplifying otherwise challenging calculations.

The Gross–Neveu–Yukawa (GNY) models serve as a rich playground for exploring these ideas. Originally introduced to describe interacting fermionic systems with spontaneous symmetry breaking, GNY models couple fermionic fields to scalar order parameters through Yukawa interactions. They are central to the study of quantum criticality, especially in lower-dimensional systems where fluctuations are enhanced. 
Remarkably, under certain conditions, such as at critical points of continuous phase transitions, GNY models can exhibit emergent SUSY. 
In these regimes, the theory becomes effectively supersymmetric, dynamically relating its fermionic and bosonic degrees of freedom. 

In this work, we develop a more rigorous framework for studying emergent SUSY, extending beyond strictly four-dimensional spacetime to regularized $D$-dimensional settings. 
Our approach constructs a unified formulation in terms of a generalized Lagrangian that encompasses a broad family of theories, including the GNY, Nambu-Jona-Lasinio-Yukawa (NJLY), and Wess–Zumino (WZ) models. 
This Lagrangian features a general number $(n_s, n_f)$ of scalar and fermion flavors, with the appropriate flavor structures, as well as the most general ternary Yukawa and quartic scalar interactions. 
By embedding all these models into a single theoretical structure, we clarify their interrelations and illuminate the mechanisms underlying emergent SUSY across different dimensions.

This unified Lagrangian also has practical applications, particularly in optimising perturbative calculations. Previously, models such as GNY and NJLY appeared unrelated; in contrast, both their supersymmetric limits and their non‑supersymmetric forms now emerge as special cases of the generalized Lagrangian. 
This structure allows to use additional SUSY Ward identities, derived from the supersymmetric limit, even when the target theory is itself non‑supersymmetric, thereby simplifying otherwise difficult Feynman diagram computations.
 This optimisation is especially useful for calculating anomalous dimensions of twist‑two operators, which play a role similar to those in quantum chromodynamics (QCD) governing the scale dependence of parton distributions. 
At high loop orders, direct computations for general spin $N$ becomes intractable, requiring many individual fixed-$N$ values to be evaluated. 
Emergent SUSY can substantially reduce the number of integrals that must be computed, making high‑loop operator renormalisation in the GNY model more tractable.

In Sec.~\ref{sec:bg}, we provide a concise introduction to the GNY, NJLY, and WZ models. Section~\ref{sec:genLagrangian} presents the generalized Lagrangian, with particular emphasis on the flavour structure, Clifford algebra, trace identities, and the regularization in $D = 4 - 2\eps$ dimensions.
All diagrammatic computations relevant for the renormalisation of the models considered are discussed and summarised in Sec.~\ref{sec:computations}. 
Based on these results, we then analyse the conditions for emergent supersymmetry as functions of $n_s$ and $n_f$.
The application to optimizing operator renormalisation is presented in Sec.~\ref{sec:omes}.
We summarize and present an outlook in Sec.~\ref{sec:concl}, while some technicalities on Ward identities and light-cone SUSY are deferred to 
Apps.~\ref{app:wi} and 
\ref{app:lightconesusy}.

\if{1=0]

\begin{itemize}
    \item Seva's thesis \cite{Chestnov:2019bed}
\end{itemize}

\fi

\section{Background}
\label{sec:bg}

\subsection{Gross-Neveu-Yukawa models}\label{gny}
The original Gross–Neveu model \cite{Gross:1974jv} was introduced as a two-dimensional theory of massless fermions with quartic self-interactions, designed to provide a simple, renormalisable setting for studying dynamical symmetry breaking in asymptotically free field theories. Owing to its relative solvability, it served as a laboratory for exploring non-perturbative phenomena, for instance the anticipated spontaneous breaking of chiral symmetry by the QCD vacuum.

The Gross-Neveu-Yukawa (GNY) model is the bosonized formulation of the Gross‑Neveu model, expressed in four dimensions, with the bare Lagrangian
\begin{equation}
    \label{eq:GNYbare}
\mathcal{L}^{\text{bare}}_{\text{GNY}} = \frac{1}{2}\partial_\mu \phi_0\partial^\mu \phi_0+ i(\bar\Psi^A)_0\slashed{\partial}(\Psi_A)_0-g_0\,\phi_0(\bar\Psi^A)_0(\Psi_A)_0 - \frac{\lambda_0}{4!}\phi_0^4
\, ,
\end{equation}
where $\phi$ is a real scalar and $\Psi_A$ is a Dirac fermion with flavour index $A$, subject to a global $SU(n_f)$ flavour symmetry. The subscripts $(_0)$ denote bare quantities. 
The two coupling constants $a$ and $\lambda$ mix under renormalisation, and are renormalised in the $\overline{\text{MS}}$-scheme via the matrix equation
\begin{equation}\label{couplings}
    \begin{bmatrix}
        \lambda_0 \\
        a_0
    \end{bmatrix} =(4\pi)^2\mu^{2\eps}
    \begin{bmatrix}
        Z_{\lambda\lambda} & Z_{a\lambda} \\
        Z_{\lambda a} & Z_{aa}
    \end{bmatrix}
    \begin{bmatrix}
        \lambda \\
        a
    \end{bmatrix}
\, ,
\end{equation}
where $a_0=g_0^2$, and $\mu$ is the usual $\overline{\text{MS}}$ renormalisation scale with the implicit $4\pi e^{-\gamma_E}$ factor.  Furthermore, 
\begin{equation}
 Z_{\lambda\lambda}\, , Z_{aa} = 1 +\mathcal{O}(\lambda,a)\, ,\qquad Z_{a\lambda} = \mathcal{O}(a)\, ,\qquad Z_{\lambda a} = \mathcal{O}(\lambda)\;,
\end{equation}
signifying that the off-diagonal terms have no tree-level contribution and are thus purely from quantum corrections.
The renormalisation group functions are defined as:
\begin{equation}
\begin{aligned}
    \gamma_\phi &= \frac{d\ln Z_\phi}{d\ln\mu}\, ,\qquad
    \gamma_\psi= \frac{d\ln Z_\psi}{d\ln\mu}\;,
    \\
    \beta_a &= \frac{da}{d\ln\mu}\, , \quad\qquad\!\!
    \beta_\lambda = \frac{d\lambda}{d\ln\mu}\;.
\end{aligned}
\end{equation}
Similar to this we have the NJLY model with complex scalars:
\begin{equation}
\label{eq.NJLYbare}
\mathcal{L}^{\text{bare}}_{\text{NJLY}} = \partial_\mu \phi^\dagger\partial^\mu \phi- i\bar\Psi^A\slashed{\partial}\Psi_A+g\,(\phi\bar\Psi^AP_+\Psi_A+\phi^\dagger\bar\Psi^AP_-\Psi_A) + \frac{\lambda}{4}(\phi^\dagger\phi)^2
\, ,
\end{equation}
with the projections $P_{\pm} = \frac{1}{2}(1\pm\gamma_5)$, and the subscripts $(_0)$ are dropped for brevity.

These models lie in the same universality class as some chiral Ising models, and hence their critical point behaviour is of much interest in condensed matter physics. 
The pertubative renormalisation has been achieved up to four loops in both models \cite{Zerf:2017zqi} and also recently at five loops for the GNY model \cite{Gracey:2025aoj}.

Previous studies \cite{Lee:2006if, Grover:2013rc, Zerf:2016fti} have shown that supersymmetry can emerge in various models through both lattice and perturbative methods. 
Specifically, in the GNY model, SUSY appears at a fermion flavor number of  $n_f=1/4$, corresponding to an $\mathcal{N}=1$ three-dimensional WZ model. 
Similarly, in the Nambu–Jona-Lasinio-Yukawa (NJLY) model, emergent SUSY is observed at $n_f=1/2$, which relates to an $\mathcal{N}=1$ four-dimensional WZ model. 
By dimensional reduction, this also connects to the $\mathcal{N}=2$ WZ model in three dimensions.

\subsection{Wess-Zumino models}

The WZ model is based on the massless chiral multiplet. 
The four-dimensional $\mathcal{N}=1$ multiplet $\{\phi,\psi,F\}$ comprises of one complex scalar and one Weyl fermion on shell. 
The infinitesimal SUSY transformations act on the fields as:
\begin{equation}
\begin{aligned}
    \delta_{\epsilon,\bar\epsilon}\,\phi &= \sqrt{2}\epsilon^\alpha\psi_\alpha\, ,\\ 
    \delta_{\epsilon,\bar\epsilon}\,\psi_{\alpha} &= \sqrt{2}i(\sigma^\mu\bar\epsilon)_{\alpha}\partial_\mu\phi -\sqrt{2}\epsilon_{\alpha}F\, ,\\
    \delta_{\epsilon,\bar\epsilon}\, F& = \sqrt{2}i\partial_\mu\psi\sigma^\mu\bar\epsilon
\, , 
\end{aligned}   
\end{equation}
and taking the complex conjugate yields the corresponding transformations of the conjugate fields. 
The fields $F,F^\dagger$ are auxilliary fields without kinetic terms, and can be expressed in terms of the other fields upon imposing equations of motion. The bare Lagrangian for the four-dimensional WZ model is then given as:
\begin{equation}\label{eq:wz}
    \mathcal{L}^{\text{bare}}_{\text{WZ}} = \partial_\mu \phi^\dagger_0\partial^\mu \phi_0+ i\bar\psi_0\bar\sigma^\mu\partial_\mu\psi_0-g_0\,\phi_0\bar\psi_0\bar\psi_0 -g_0\,\phi^\dagger_0\psi_0\psi_0 -g^2_0(\phi^\dagger_0\phi_0)^2
\, ,
\end{equation}
which transforms as a total derivative under a SUSY transformation, and therefore leaves the action invariant.
Notice that, because of SUSY both the scalar and Yukawa vertices run identically with $g_0$, while the anomalous dimensions of the fields are equal,
\begin{equation}\label{susyanomf}
    \gamma_\phi= \gamma_\psi\, .
\end{equation}
The fields $\{\phi,\psi,F\}$ can be combined into one chiral superfield $\Phi$, which reads:
\begin{equation}
  \Phi(x,\theta, \bar\theta) = \phi -\sqrt{2}\theta\psi +i\theta\sigma^\mu\bar\theta\partial_\mu\phi - \theta\theta F -\frac{i}{\sqrt{2}}\theta\theta\partial_\mu\psi\sigma^\mu\bar\theta - \frac{1}{2}\theta\theta\bar\theta\bar\theta\Box\phi  
\, .
\end{equation}
Similarly, the anti-chiral field $\bar\Phi$  is obtained by taking the Hermitian conjugate. 
Then the Lagrangian in eq.~\eqref{eq:wz} is given by the integration over superspace of the following superfield Lagrangian:
\begin{equation}
    \int d^2\theta\, d^2\bar\theta\;\bar\Phi\Phi - \int d^2\theta \;g\Phi^3 - \int d^2\bar\theta \;g\bar\Phi^3 
\end{equation}
up to possible total derivatives. 
The superpotential, that gives us the interaction terms in $\mathcal{L}_{\text{WZ}}$, is defined as 
\begin{equation}
V(\Phi)=g\Phi^3\, ,    
\end{equation}
and, similarly, for the Hermitian conjugate. 
A non‑renormalisation theorem holds for the superpotential constrained by holomorphy \cite{Grisaru:1979wc}, which implies that 
\\
\begin{equation}\label{susyanomb}
    Z_gZ^{3/2}_\phi = 1\implies\beta_a = -3\gamma_\phi = -3\gamma_\psi
\, ,
\end{equation}
where $a=g^2$ like the previous section. 
We also encounter the three-dimensional WZ model as an emergent SUSY of the GNY Lagrangian, where the multiplet $\{\varphi,\chi,F\}$ consisting of a real scalar and a three-dimensional Majorana fermion on shell: 
\begin{equation}\label{eq:3dwz}
    \mathcal{L}^{\text{bare}}_{\text{3D-WZ}} = \frac{1}{2}\partial_\mu \varphi_0\partial^\mu \varphi_0+\frac{1}{2}i\bar\chi_0\gamma^\mu\partial_\mu\chi_0-\frac{1}{2}g_0\,\varphi_0\bar\chi_0\chi_0 -\frac{1}{8} g^2_0\varphi^4
\, .
\end{equation}
Here the fermions and $\gamma$ matrices all correspond to their three-dimensional representations. 
The corresponding anomalous dimensions of the fields are equal and $\beta$ functions from the Yukawa and quartic vertices run identically. However, due to the lack of holomorphy, the corresponding non-renormalisation theorem for the superpotential does not exist. Hence eq.~\eqref{susyanomf} holds but eq.~\eqref{susyanomb} does not hold in the three-dimensional WZ model.

\section{The Generalised Lagrangian}
\label{sec:genLagrangian}

Here, we aim to derive a generalised Lagrangian that unifies the GNY, NJLY, and WZ models. This Lagrangian should contain an arbitrary number $(n_s,n_f)$ of scalar and fermion flavours, together with the appropriate flavour structures. We begin by writing the most general form of a theory involving scalars and fermions, including a ternary Yukawa interaction and a quartic scalar coupling:
\begin{align}
\label{lgen}
 \mathcal{L}_{\text{gen}} &= \frac{1}{2}\partial_\mu \phi^I\partial^\mu \phi_I+ i\bar\psi^A\bar\sigma^\mu\partial_\mu\psi_A-g(\bar Z^I)^{AB}\phi_I\psi_A\psi_B - g( Z^I)_{AB}\phi_I\bar\psi^A\bar\psi^B 
 \nonumber \\& \qquad\qquad
 - \frac{\lambda}{4!}M^{IJKL}\phi_I\phi_J\phi_K\phi_L
\, ,
\end{align}
where $\phi_I$ are real scalars carrying a flavour index transforming under the fundamental representation of a global $SO(n_s)$ flavour symmetry. 
Since this is a real representation, the raised field $\phi^I$ transforms in the same representation. 
Indices are raised and lowered with the invariant tensor $\delta^{IJ}$ and $\delta_{IJ}$, and the identity operator is given by $\delta^I_{\ J}=\delta^{IK}\delta_{KJ}$. 

$\psi_A$ is a Weyl spinor transforming under a fundamental global $SU(n_f)$ flavour symmetry, 
while its conjugate field $\bar\psi^A$ transforms in the anti‑fundamental representation.
Since this is a complex representation, there is in general no invariant tensor that raises or lowers flavour indices. The only invariant tensor is the identity $\delta^A_{\ B}$, where the raised index $A$ is anti‑fundamental and the lowered index $B$ is fundamental.

$(\bar Z^I)^{AB}, (Z^I)_{AB}$,  and $M^{IJKL}$ are \emph{flavour elements}~\footnote{Note the distinction between flavour elements and flavour factors: the former are indexed entities while the latter are complex numbers obtained from traces of flavour elements.}. This generalised form of writing the Lagrangian is similar in appearance to previous works in the literature. For example, our flavour elements are analogous to the tensor structures used in \cite{Poole:2019kcm}. However, this is where our approach diverges from the literature: previous applications use an explicit matrix structure of the flavour elements for computation. In contrast, we keep the flavour elements abstract, and only impose minimal algebraic relations on them during computation. This is what allows us the flexibility to move across different models. 

At this stage, no algebraic constraints are yet imposed, except for the index structure of the flavour elements, indicating how they transform under the flavour symmetries. The Lagrangian is therefore a flavour singlet, since all flavour indices are fully contracted.

For Weyl spinors, we have $\psi_A\psi_B=\psi_B\psi_A$ and $\bar\psi^A\bar\psi^B=\bar\psi^B\bar\psi^A$, which implies 
\begin{equation}
(\bar Z^I)^{AB} = (\bar Z^I)^{BA}\quad,\quad ( Z^I)_{AB} = (Z^I)_{BA}\, .
\end{equation}
Similarly, the tensor $M^{IJKL}$ is fully symmetric under permutations of its indices.  
Hermiticity of the Lagrangian then imposes 
\begin{equation}
(\bar Z^I)^{AB} = \left((Z^I)_{BA}\right)^*\quad,\quad M_{IJKL}=M^\ast_{IJKL}\, .
\end{equation}
Thus, the flavour elements $Z^I, \bar Z^I$ are symmetric and Hermitian conjugates of each other, while $M^{IJKL}$ is real and symmetric in the scalar flavour indices.

\subsection{The flavour elements}

We now turn to the algebraic structure of the flavour elements. As a starting point, we consider the properties of the generalised Lagrangian.
We require the absence of \emph{evanescent flavour terms}~\footnote{Evanescent terms in a Lagrangian are those that do not exist in the classical case, but are added by quantum corrections in order to cancel divergences from loop diagrams.} 
in $\mathcal{L}_{\text{gen}}$. 
This requirement is motivated by the fact that none of our target theories (GNY, NJLY, or WZ) contain such terms. 
Consequently, the ultraviolet counterterms for propagators and vertices (with fermions represented by solid lines and scalars by dashed lines) must carry the same flavour structure as the corresponding tree‑level diagrams:
\begin{equation}\label{feynmanrules}
\begin{split}
    \glsclr \propto\delta^{IJ} \quad \glfer\propto&\delta_{A}^{\ B} \quad \glfour\propto M^{IJKL}\\
     \glyuk \propto(Z^I)^{AB} \quad\quad &\glyukconj \propto(\bar Z^I)_{AB}
\end{split}
\end{equation}
Note above that the fermion arrows converge at the Yukawa vertex rather than flowing along the same direction as is usually known. This is because we are using Weyl fermions instead of Dirac fermions.

It is sufficient to establish non‑evanescence only for the primitive diagrams of the theory, i.e. diagrams without sub‑divergences. This guarantees non‑evanescence at all loop orders, since every divergent correction to propagators or vertices contains primitive diagrams as subgraphs.

The requirement that no evanescent terms appear, therefore, imposes constraints on the flavour elements $(Z^I)_{AB}$ and $M^{IJKL}$. 
This gives rise to the following theorem:
\begin{theo}\label{thm:nonevanescence}
There are no evanescent flavour terms in the generalised Lagrangian $\lgen$ if and only if the $Z,\bar Z$ satisfy the following relation:
\begin{equation}
(Z^I)_{AC}(\bar Z^J)^{CB} + (Z^J)_{AC}(\bar Z^I)^{CB} = 2\dij\delta_A^{\ B}
\;.    
\end{equation}
The $Z^I,\bar Z^I$ are thus representations (possibly reducible) of the modified Clifford Algebra
\begin{equation}\label{eq:flavclifford}
    Z^I\bar Z^J+Z^J\bar Z^I = 2\delta^{IJ}\;,
\end{equation}
i.e. the same one satisfied by the Pauli matrices $\sigma^\mu,\bar\sigma^\mu$.
\end{theo}
\begin{proof}
The fact that eq.~\eqref{eq:flavclifford} is sufficient to ensure non‑evanescence is already well established in theories such as QCD and quantum electrodynamics (QED), formulated in the Weyl basis, where the Pauli matrices satisfy the modified Clifford algebra and do not generate any evanescent terms. 
One may also verify this explicitly by checking the primitive diagrams individually; however, we omit this here, as it is a straightforward exercise based on a well‑known result.

Here we address the more involved question of showing that eq.~\eqref{eq:flavclifford} is not only sufficient, but also necessary for non‑evanescence in~$\lgen$; that is, the modified Clifford algebra is the only structure that allows non‑evanescent flavour behaviour. 
Recall that any object with two indices can be decomposed as follows:
\begin{equation}\label{eq:tracedecomp}
    T^{IJ} = \tr[T]\delta^{IJ} + S^{IJ}+A^{IJ}
\, ,
\end{equation}
where $S^{IJ}$ is symmetric in $I,J$ and traceless, and $A^{IJ}$ is antisymmetric in $I,J$. Hence, we can write
\begin{equation}\label{eq:tracecliffdecomp}
    \begin{split}
    (Z^I)_{AC}(\bar Z^J)^{CB}&= t_A^{\ B}\dij + s_A^{\ B}S^{IJ} + a_A^{\ B}A^{IJ} \\
    &=t^{IJ}\delta_A^{\ B} + s^{IJ}S_A^{\ B} + a^{IJ}A_A^{\ B}
\, ,
\end{split}
\end{equation}
where in the first line eq.~\eqref{eq:tracedecomp} is applied to the scalar indices, and in the second line to the fermion indices. 
Note that the elements with the lower case letters $t,s,a$ are not necessarily symmetric or anti-symmetric. 
Also, owing to the symmetry of $Z^I,\bar Z^I$, we have:
\begin{equation}\label{eq:transpose}
    (Z^I)_{AC}(\bar Z^J)^{CB} = (Z^I)_{CA}(\bar Z^J)^{BC} = (\bar Z^J)^{BC}(Z^I)_{CA}\, .
\end{equation}
Thus, interchanging $I\leftrightarrow J$ in eqs.~\eqref{eq:tracecliffdecomp} and \eqref{eq:transpose}, we get
\begin{equation}\label{eq:tracecliffdecomp2}
    \begin{split}
    (\bar Z^I)^{BC}(Z^J)_{CA}&= t_A^{\ B}\dij + s_A^{\ B}S^{IJ} - a_A^{\ B}A^{IJ} \\
    &=t^{JI}\delta_A^{\ B} + s^{JI}S_A^{\ B} + a^{JI}A_A^{\ B}
\, .
\end{split}
\end{equation}
Note that the index structure appears as $\delta^A_{\ B}$, rather than $\delta_B^{\ A}$, and analogously for $S_A^{\ B}$ and $A_A^{\ B}$.  
To ensure the absence of evanescent terms at one loop in the fermion propagator (solid lines with arrows in direction of fermion flow), it is necessary that:
\bnu\label{lhsfer}
\begin{aligned}
 \olfer \propto \delta_A^{\ B}\; &\implies\; (Z^I)_{AC}(\bar Z^I)^{CB} \,\propto \delta_A^{\ B}
\, .
\end{aligned}
\enu
For the one-loop scalar propagator (dashed lines) with two diagrams, we have
\bnu\label{lhsscl}
\begin{aligned}
 \olsclro + \olsclrt \propto \dij\; &\implies\; (Z^I)_{AB}(\bar Z^J)^{BA} +(Z^J)_{AB}(\bar Z^I)^{BA}\,\propto \dij\, .\\
\end{aligned}
\enu
Contracting both expressions in eq.~\eqref{eq:tracecliffdecomp} with the corresponding invariant tensors yields:
\bnu \label{rhs}
\begin{aligned}
    (Z^I)_{AC}(\bar Z^J)^{CB}\dij &= (Z^I)_{AC}(\bar Z^I)^{CB} = n_st_A^{\ B} \, ,\\
    \left((Z^I)_{AC}(\bar Z^J)^{CB}+(Z^J)_{AC}(\bar Z^I)^{CB}\right)\delta^B_{\ A} &= (Z^I)_{AB}(\bar Z^J)^{BA} +(Z^J)_{AB}(\bar Z^I)^{BA} \\ 
    &= n_f(t^{IJ}+t^{JI}) 
    \, . 
\end{aligned}    
\enu
Comparing eqs.~\eqref{lhsfer}, \eqref{lhsscl} with \eqref{rhs} we get: 
\begin{equation}
t^{IJ}+t^{JI}=2\delta^{IJ}\;,
\qquad 
t_A^{\ B}=\delta_A^{\ B}\: .  
\end{equation}
The latter relation implies that  $s_A^{\ B}$ and $a_A^{\ B}$ are both traceless, since $t_A^{\ B}$ entirely captures the trace component. 
Now consider the following relations:
\begin{equation}\label{antic}
\begin{split}
    (Z^I)_{AC}(\bar Z^J)^{CB} + (Z^J)_{AC}(\bar Z^I)^{CB} = 2\dij\delta_A^{\ B}+2S^{IJ}s_A^{\ B}\;,\\[3pt]
    (\bar Z^I)^{BC}(Z^J)_{CA} + (\bar Z^J)^{BC}(Z^I)_{CA} = 2\dij\delta_A^{\ B}+2S^{IJ}s_A^{\ B}\;.
\end{split}
\end{equation}
We now demonstrate that the term $S^{IJ}s_{AB}$ must vanish in order to ensure the absence of evanescent terms.  
From the one‑loop vertex correction, we obtain:
\begin{equation}\label{rhsyuk}
    \olyuk \propto\; (Z^I)_{AB} \implies (Z^J)_{AC}(\bar Z^I)^{CD}(Z^J)_{DB}\,\propto\, (Z^I)_{AB}\;.
\end{equation}
Using eq.~\eqref{antic}, we obtain
\begin{align*}
    (Z^J)_{AC}(\bar Z^I)^{CD}(Z_J)_{DB} & = -(Z^I)_{AC}(\bar Z^J)^{CD}(Z_J)_{DB}+ 2\dij \delta_A^{\ D}(Z_J)_{DB}+2S^{IJ}s_A^{\ D}(Z_J)_{DB}\\
    &=-n_s(Z^I)_{AC}\delta^C_{\ B} + 2\dij \delta_A^{\ D}(Z_J)_{DB}+2S^{IJ}s_A^{\ D}(Z_J)_{DB}\\
    &= (2-n_s)(Z^I)_{AC} + 2S^{IJ}s_A^{\ D}(Z_J)_{DB}\;.
\end{align*}
For non-evanescence we need 
\begin{equation}\label{eq:yuknonevans}
 S^{IJ}s_A^{\ D}(Z_J)_{DB} \propto Z^I_{AB} =\delta^{IJ}\delta_A^{\ D}(Z_J)_{DB}\, ,   
\end{equation}
but this relation alone is not sufficient to guarantee the vanishing of the term, since the scalar and fermion indices mix in this expression, making it non‑trivial to impose linear independence.  
The other Yukawa primitive diagram (shown below) gives us nothing new, but the same condition \eqref{eq:yuknonevans} for non-evanescence.
\begin{equation*}
    \tlyuk
\end{equation*}
Hence, we have to shift our focus to the quartic vertex primitive graphs. 
The quartic flavour structure $M^{IJKL}$ can be related to the $Z^I,\bar Z^I$ through the following diagrams:
\begin{equation}
\olfouro + \olfourt + \;(J\leftrightarrow K)\; +\; (J\leftrightarrow L)\;\propto\; M^{IJKL}\;.
\end{equation}
Using eq.~\eqref{antic} and the fact that $\tr[s]=0$, we get:
\begin{equation}\label{mijklwiths}
    M^{IJKL} \propto n_f(\delta^{IJ}\delta^{KL}+\delta^{IK}\delta^{JL}+\delta^{IL}\delta^{JK})+\tr[s^2](S^{IJ}S^{KL}+S^{IK}S^{JL}+S^{IL}S^{JK})\;,
\end{equation}
where the proportionality constant will be determined later. Substituting the above expression for $M^{IJKL}$ in the following primitive diagrams gives us:
\begin{align}
    \olfourscl + & \;(J\leftrightarrow K)\; +\; (J\leftrightarrow L) \nonumber \\ &\propto n_f^2(n_s+2)\delta^{IJ}\delta^{KL}  + (\tr[s^2])^2\left(\tr_I[S^2]S^{IJ}S^{KL} + 2(S^2)^{IJ}(S^2)^{KL}\right)\nonumber \\
    &+(J\leftrightarrow K) +(J\leftrightarrow L) 
    \, ,
    \label{eq:oneloopscl}
\end{align}
where $\tr_I$ indicates a trace over the scalar flavour indices. 
For the above expression to be proportional to $M^{IJKL}$, we must have
\begin{equation}
    S^2 \propto S,\quad \text{but } \tr_I[S]=0 \implies\tr_I[S^2]=0
    \, .
\end{equation}
However, since $S$ is real and symmetric, the matrix $S^2$ is diagonalizable and its eigenvalues are non‑negative. 
The tracelessness of $S^2$ therefore implies that all of its eigenvalues, and hence all eigenvalues of $S$, must vanish. Consequently, we obtain:
\begin{equation}
S^{IJ} = 0\;.
\end{equation}
Hence, the flavour elements $Z^I$ and $\bar Z^I$ must necessarily satisfy the modified Clifford algebra in order to ensure non‑evanescence for all propagators and vertices:
\begin{equation}\label{cliffindices}
    \begin{split}
    (Z^I)_{AC}(\bar Z^J)^{CB} + (Z^J)_{AC}(\bar Z^I)^{CB} = 2\dij\delta_A^{\ B}\;,\\[3pt]
    (\bar Z^I)^{AC}(Z^J)_{CB} + (\bar Z^J)^{AC}(Z^I)_{CB} = 2\dij\delta_B^{\ A}\;.
\end{split}
\end{equation}
which completes our proof.
\end{proof}
\begin{cor}
    Under non-evanescence, the quartic flavour element is given by 
    \begin{equation}\label{mijkl}
    M^{IJKL} = (\delta^{IJ}\delta^{KL}+\delta^{IK}\delta^{JL}+\delta^{IL}\delta^{JK})
\end{equation}
where the proportionality constant is chosen so that the Feynman rule for the quartic vertex is $-i\lambda M^{IJKL}$.
\end{cor}
\begin{proof} 
Substituting $S=0$ into eq.~\eqref{mijklwiths} yields the above expression for $M^{IJKL}$. 
With this form of $M^{IJKL}$, eq.~\eqref{eq:oneloopscl} becomes non‑evanescent. 
It is a straightforward algebraic exercise to verify that all other primitive diagrams are likewise non‑evanescent with this definition of $M^{IJKL}$. 
Our four‑loop computations confirm this explicitly. We therefore omit the detailed derivation here.
\end{proof}
These conditions on the flavour elements $\zab$ and $M^{IJKL}$ completely determine the values of all traces occuring in the one‑particle‑irreducible (1PI) correlators for the propagators and vertices. 
Thus, all flavour factors of all diagrams are fixed by imposing the single requirement that no evanescent flavour elements appear. 
In the next section, we derive the relevant trace identities and discuss the subtleties associated with regularisation.

\subsection{Trace identities}

In this section we derive the relevant identities for traces of the flavour elements $Z^I$ and $\bar Z^I$. 
The derivations presented here apply to positive integer dimensions; the extensions required for dimensional regularisation will be discussed in the following section.

\subsubsection*{Even trace identities} 

First, let us compute the trace of an even number of  $Z$ and  $\bar Z$. 
Note that the traces encountered in the evaluation of Feynman diagrams involve alternating $Z$ and  $\bar Z$ elements. 
Consequently, we are interested in traces of the form
\begin{equation}
    \tr[Z^{I_1}\bar Z^{I_2}\ldots \bar Z^{I_{2n}}]\quad,\quad\tr[\bar Z^{I_1} Z^{I_2}\ldots Z^{I_{2n}}]\;.
\end{equation}
We can combine the $Z^I$ and $\bar Z^I$ into one entity:
\begin{equation}\label{eq:gamma}
    \Gamma^I = \begin{pmatrix}
    0 & Z^I\\\bar Z^I&0
\end{pmatrix}
\, .
\end{equation}
The $\Gamma^I$ then form representations of the Clifford algebra:
\begin{equation}\label{eq:GammaClifford}
    \Gamma^I\Gamma^J+\Gamma^J\Gamma^I = 2\delta^{IJ}\;.
\end{equation}
Then, the trace of an even number of $\Gamma^I$ is given by:
\begin{equation}\label{eq:achiral}
    \tr[\Gamma^{I_1}\Gamma^{I_2}\ldots\Gamma^{I_{2n}}] = \tr[Z^{I_1}\bar Z^{I_2}\ldots \bar Z^{I_{2n}}]+\tr[\bar Z^{I_1} Z^{I_2}\ldots Z^{I_{2n}}]
\, .
\end{equation}
As is well known, the trace above can be derived solely from the anti‑commutation relation, and is therefore independent of the specific representation:
\begin{equation}\label{eq:tracegamma}
    \tr[\Gamma^{I_1}\Gamma^{I_2}\ldots\Gamma^{I_{2n}}] = \sum_{j=2}^{2n}(-1)^{j}\delta^{I_1I_{j}}\tr[\Gamma^{I_2}\ldots \Gamma^{I_{j-1}}\Gamma^{I_{j+1}} \Gamma^{I_{2n}}]
\, ,
\end{equation}
where the solution of the recursion yields a sum of products of Kronecker deltas. 
But this is only a sum of the traces of $Z^I$ and $\bar Z^I$. 
In order to isolate the individual traces, we require the following ingredient:
\begin{equation}\label{eq:gamma5proj}
    \Gamma_{(n_s+1)} = \begin{pmatrix}
    \ide_Z & 0\\0&-\ide_Z
\end{pmatrix}
\, ,
\end{equation}
where $\ide_Z$ denotes the identity element of the same dimension as  $Z^I$ and  $\bar Z^I$. 
This yields:
\begin{equation}\label{eq:tracesig}
\begin{split}
    \tr[Z^{I_1}\ldots \bar Z^{I_{2n}}]= \frac{1}{2}\tr[\Gamma^{I_1}\ldots\Gamma^{I_{2n}}(1+\Gamma_{(n_s+1)})]\, ,\\
     \tr[\bar Z^{I_1}\ldots Z^{I_{2n}}]= \frac{1}{2}\tr[\Gamma^{I_1}\ldots\Gamma^{I_{2n}}(1-\Gamma_{(n_s+1)})]
     \, .
\end{split}
\end{equation}
For the case when $n_s$ is a positive even integer, we can express $\Gamma_{(n_s+1)}$ in terms of $\Gamma^I$ as:
\begin{equation}\label{eq:gamma5}
\Gamma_{(n_s+1)} =
     \dfrac{i^{\frac{n_s}{2}}}{n_s!}\epsilon_{I_1\ldots I_{n_s}}\Gamma^{I_1}\ldots\Gamma^{I_{n_s}} 
\, ,
\end{equation}
where the overall normalisation accounts for Euclidean metric of the $\delta^{IJ}$. 
This is analogous to $\gamma_5$, hence the subscript $(n_s+1)$. 
With the above definition, we can calculate the right hand side of eq.~\eqref{eq:tracesig} using eq.~\eqref{eq:tracegamma}. 
For example,
\begin{equation}\label{eq:neq2}
    \tr[Z^I\bar Z^J] = 
    \begin{cases}
        \tr[\ide_Z](\delta^{IJ}+i\epsilon^{IJ})&\text{for $n_s = 2$} \\
        \tr[\ide_Z]\delta^{IJ}&\text{for $n_s > 2$}
    \end{cases}
    \, ,
\end{equation}
while
\begin{equation}\label{eq:neq4}
    \tr[Z^I\bar Z^JZ^K\bar Z^L] = 
    \begin{cases}
        \tr[\ide_Z]\left((\delta^{IJ}\delta^{KL}-\ldots)+i(\epsilon^{IJ}\delta^{KL}-\epsilon^{IK}\delta^{JL}+\ldots)\right)&\text{for $n_s = 2$} \\
        \tr[\ide_Z]\left((\delta^{IJ}\delta^{KL}-\ldots)-\epsilon^{IJKL}\right)&\text{for $n_s = 4$}\\
        \tr[\ide_Z](\delta^{IJ}\delta^{KL}-\ldots)&\text{for $n_s > 4$}
    \end{cases}
    \, .
\end{equation}
Notice the sequential ``unlocking" of Levi‑Civita contributions as the number of $Z^I$ factors in the trace  increases. 
For $n_s=4$, the Levi‑Civita term contributes only when the trace contains four or more  $Z^I$ matrices; for $n_s=6$, the corresponding threshold is six. 
This occurs because a sufficient number of $\Gamma^I$ matrices are required on the right‑hand side of eq.~\eqref{eq:tracesig} in order to produce a non‑vanishing Levi‑Civita term.

We refer to the component of $\tr[Z^{I_1}\ldots \bar Z^{I_{2n}}]$ that is independent of the Levi‑Civita tensor as the \emph{achiral} part, while the component proportional to the Levi‑Civita tensor is termed the \emph{chiral} part.

\subsubsection*{Odd trace identities}

Multiplying the $\Gamma$ matrices an odd number of times always produces a product of $Z$ matrices in the off‑diagonal blocks, which is not useful for deriving the traces with an odd number of $Z^I$ and $\bar Z^I$:
\begin{equation*}
    \tr[Z^{I_1}\bar Z^{I_2}\ldots Z^{I_{2n-1}}]\quad,\quad\tr[\bar Z^{I_1} Z^{I_2}\ldots \bar Z^{I_{2n-1}}]\;.
\end{equation*}
There is no known general relation for traces involving an odd number of $Z$ matrices. 
Recall also that, since we work with Weyl fermions, the propagators contain $\sigma^\mu$ and $\bar \sigma^\mu$, which constitute the four‑dimensional representation of $Z^I$ and $\bar Z^I$, 
and therefore exhibit the same limitation.

Fortunately, the combinatorics of Feynman graphs in an even‑dimensional quantum field theory render the evaluation of odd traces unnecessary. A trace involving an odd number of
$\sigma^\mu,\bar \sigma^\mu$ or $Z^I,\bar Z^I$ could only arise from a Weyl‑fermion loop containing an odd number of Yukawa vertices. However, such a diagram cannot occur: the structure of the Yukawa vertex, as depicted in eq.~\eqref{feynmanrules}, requires the fermion arrow to flip at each vertex. A closed fermion loop requires an even number of arrow flips, and hence an even number of Yukawa vertices.

\subsubsection*{Reducible representations}

In certain theories, the representation of the $Z$ matrices may be reducible—for example, in $\mathcal{N}=2$ super Yang–Mills (SYM) theory. 
It is therefore necessary to address whether the trace identities derived above differ between reducible and irreducible representations.

Since odd traces never appear, only the even traces are relevant. 
The traces of an even number of $Z^I$ and $\bar Z^I$ are expressed entirely in terms of the even traces of the $\Gamma^I$ matrices, eq.~\eqref{eq:tracegamma}, which follow solely from the Clifford algebra and therefore do not depend on the choice of representation. Consequently, all trace identities derived here hold equally in irreducible and reducible representations. 

Linear representations of $Z^I$ and $\bar Z^I$ exist only for positive even integer values of $n_s$, and the dimension of the corresponding irreducible representations is $2^{(n_s/2)-1}$. 
Therefore, the cases with odd $n_s$  must be treated separately.

\subsubsection*{Odd dimensions}

We will restrict ourselves to the $n_s=1$ case as it is the only odd case relevant to this work. In this case we can incidentally find a representation of $Z^I,\bar Z^I$. To see this, consider the irreducible representation of the Clifford algebra $n_s=1$, which is one-dimensional and real, hence:
\begin{equation}
\left((Z^I)_{AB}\right)^* = (Z^I)_{AB} \implies (\bar Z^I)^{AB} = (Z^I)_{AB}
\, .
\end{equation}
Then the eq.~\eqref{eq:flavclifford} becomes a Clifford algebra, and therefore
\begin{equation}
Z^I=\Gamma^I
\, .
\end{equation}
The $n_s=1$ case is trivial as the $\Gamma^I=\ide$, but in general, the trace of an odd number of $\Gamma^I$ in odd dimensions is not zero.
This is relevant to the emergent SUSY at $n_s=1$ which corresponds to a chiral superfield theory in three space-time dimensions~\cite{Grover:2013rc}.

The following inductive identity can be proven with just the Clifford algebra anti-commutation: (the indices are Lorentz indices, so the dimension parameter is $D$ and not $n_s$) 
\begin{equation}\label{eq:odddimtrace}
\tr[\Gamma^{\mu_1}\Gamma^{\mu_2}\ldots\Gamma^{\mu_{2n-1}}]=0 \implies (D-2n+1)\tr[\Gamma^{\mu_1}\Gamma^{\mu_2}\ldots\Gamma^{\mu_{2n+1}}]=0\, .
\end{equation}
If the parameter $D$ is not an odd number, the induction never fails and traces of any arbitrary odd number of $\Gamma$ matrices are zero. However, if $D$ is a positive odd integer, the induction fails beyond $n=(D+1)/2$. A trace of $D$ or more odd $\Gamma$ matrices is then no longer zero. For instance, in $D=3$ with a Euclidean metric, 
\begin{equation}\label{eq:oddtrace}
    \tr[\gamma^\mu\gamma^\nu\gamma^\rho]=-i\epsilon^{\mu\nu\rho}\tr[\ide_{2\times 2}]
\, ,
\end{equation}
where $\gamma^\mu$ are representations of the three-dimensioal Clifford algebra.

The graph combinatorics created by Weyl spinors is not applicable in odd dimensions. 
As discussed above, Weyl spinors cannot produce graphs with an odd number of traces. 
However, since Weyl spinors simply do not exist in odd-dimensional spacetimes, this type of graph‑combinatorial reasoning is not available there. 
As a consequence, obtaining the emergent three-dimensional SUSY with $n_s=1$ from the $\lgen$ theory at four loops becomes problematic, because our calculations rely on Weyl spinors in four dimensions.

This issue was already observed in \cite{Zerf:2017zqi}, where it was traced back to the fact that the graphs of the type shown in Figure \ref{fig:oddtrace} do not exist in four dimensions.
\begin{figure}
\centering
\includegraphics[width=0.33\linewidth]{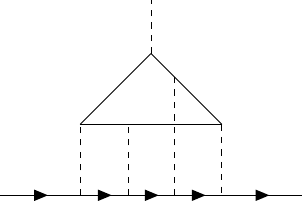}
\caption{\label{fig:oddtrace} 
Sample topologies of diagrams that yield a non‑vanishing odd trace contribution in three dimensions, with three-dimensional Dirac fermions (solid lines, arrows indicate fermion flow, no arrows indicate fermion flow in both directions) and scalars (dashed lines).
\vspace*{-1mm}}
\end{figure}
To restore emergent SUSY, one simply needs to include the missing contributions by making use of eq.~\eqref{eq:oddtrace}, as demonstrated in \cite{Zerf:2017zqi}. 
Note that, what we have in Figure \ref{fig:oddtrace} are no longer Weyl spinors, but three dimensional Dirac spinors, that actually represent three-dimensioal Majorana spinors when we substitute $n_f=1/2$ to obtain the result for the SUSY model.
\\

\subsection{Regularising to complex $D=4-2\eps$}\label{dimregsec}

In the previous section we derived the trace identities in positive integer dimensions; we now discuss how these quantities extend under regularisation. 
Dimensional regularisation~\cite{tHooft:1972tcz,Bollini:1972ui} in $D=4-2\eps$ dimensions analytically continues tensor and spinor structures to non‑integer dimensions, which can subtly alter algebraic identities that hold only in fixed integer dimensions.
In supersymmetric theories such as SYM, the scalar flavour number must be regularised as $n_s\to n_s+2\eps$ in order to preserve SUSY in $D$ dimensions. 
In the WZ model, the case relevant for this paper, $n_s$ does not require such regularisation. 
However, the Lorentz indices carried by  $\sigma^\mu$ and $\bar\sigma^\mu$ are still $D$-dimensional even in the WZ model, and therefore we must determine how the trace identities are modified under dimensional regularisation.

\paragraph{Odd traces} 
Notice that the induction in eq.~\eqref{eq:odddimtrace} does not terminate unless $D$ is a positive odd integer. 
Hence, odd traces must vanish for regularised (non‑integer) $n_s$ and $D$. 
The fact that these traces become non‑zero exactly in odd integer dimensions can be understood as a discontinuity in the analytic continuation: to regularise an odd‑dimensional theory, one must employ a different analytic continuation than the one used in even dimensions. 
For example, to obtain the three-dimensional emergent SUSY at $n_s=1$, 
we consider an analytic continuation in which odd traces in $D=3-2\eps$ dimensions do not vanish as $\eps\to 0$.
Thus,
\begin{equation}\label{eq:oddtracethree}
    \tr[\gamma^\mu\gamma^\nu\gamma^\rho]=-i\epsilon^{\mu\nu\rho}\tr[\ide_{2\times 2}] + \mathcal{O}(\eps),
\end{equation}
where the indices of the Levi–Civita tensor are defined strictly in three dimensions~\footnote{ Note the difference in symbols used for the Levi-Civita tensor ($\epsilon$) and the dimensional‑regularisation parameter ($\eps$).}. 
By contrast, when regularising an even‑dimensional theory, one may simply use the fact that the trace of an odd number of $\Gamma$ matrices vanishes.

\paragraph{Even traces} 
As the trace identity in eq.~\eqref{eq:tracegamma} follows solely from the Clifford‑algebra anti‑commutation relations, it remains valid in $D$ dimensions -- a fact well established in theories such as QCD at high loop orders. Consequently, the achiral part (the contribution independent of $\Gamma_{(n_s+1)}$) of the traces in eq.~\eqref{eq:tracesig} extends without obstruction to complex values of $D$ and $n_s$.

For the chiral part, which contains the Levi–Civita contributions, we encounter the well‑known  $\gamma_5$ problem. Because we treat $\Gamma_{(n_s+1)}$ as the matrix in eq.~\eqref{eq:gamma5proj} in order to project onto the desired trace $\tr[Z\ldots\bar Z]$ or $\tr[\bar Z\ldots Z]$, we implicitly impose that $\Gamma_{(n_s+1)}$ anticommutes with all $\Gamma^I$. Thus we are not working in the HVBM scheme \cite{tHooft:1972tcz,Breitenlohner:1977hr}, but rather in a scheme in which cyclicity of the trace does not hold, such as the reading point method proposed in \cite{Korner:1991sx}.

\begin{figure}
\centering
\includegraphics[width=0.8\linewidth]{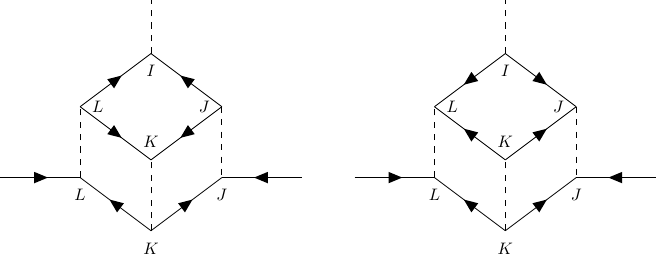}
\caption{\label{fig:anom} 
Example of the reading point issue with Dirac fermions (solid lines, arrows indicate fermion flow) and scalars (dashed lines), see footnote~\ref{ft:scalars}. 
The letters indicate the flavour indices from scalars. In the regularised case, the flavour traces can change depending on the choice of the initial point in the fermion loop.
\vspace*{-1mm}}
\end{figure}
To see that, consider the diagrams in Figure \ref{fig:anom} for $n_s=2$.
While computing the 1PI vertex correlator, we apply the Feynman rules and contract the vertex corrections with $\bar Z^I_{AB}$. 
Let us now examine the flavour factor arising from the sum of these two graphs~\footnote{\label{ft:scalars}
The diagrams derive from the generalised Lagrangian which has real scalars, so we do not have arrows on the scalar propagators. 
The component WZ Lagrangian derived from the two chiral superfields, for example in~\cite{Avdeev:1982jx}, 
has complex scalars $\phi$ and $\bar \phi$, and hence needs directed propagators.
}, 
based on the trace identity \eqref{eq:neq4} :
\begin{equation}
    \left(\tr[Z^{I}\bar Z^{J}Z^{K}\bar Z^{L}] +\tr[ \bar Z^{I}Z^{J}\bar Z^K Z^{L}]\right)\times \tr[\bar Z^{I} Z^{J}\bar Z^{K} Z^{L}] = 0\Big |_{\text{Levi-Civita}}
    \, ,
\end{equation}
where we are displaying only the contributions coming from the Levi–Civita tensors. We could equally well apply the Feynman rules starting from the vertex $J$ in the fermion loop on the right instead of $I$. 
In that case the result using the same trace identity \eqref{eq:neq4} is:
\begin{equation*}
    \left(\tr[Z^{I}\bar Z^{J}Z^{K}\bar Z^{L}] +\tr[ Z^{J}\bar Z^K Z^{L}\bar Z^{I}]\right)\times \tr[\bar Z^{I} Z^{J}\bar Z^{K} Z^{L}] \propto n_s(n_s-1)(n_s-2)\Big |_{\text{Levi-Civita}}
    \, .
\end{equation*}
This is still zero as we have $n_s=2$, but in theories such as $\mathcal{N}=2$ SYM, where $n_s=2+2\eps$, the same diagrams yield different results if one simply reorders the application of the Feynman rules. This signals an inconsistency. 
The standard workaround is the reading‑point method, in which the initial point of the trace is fixed \cite{Korner:1991sx}. In the case of SYM, we know that the correct Levi–Civita contribution must vanish in order to obtain SUSY‑consistent results. Consequently, the appropriate reading point is at $I$, which was used in~\cite{n2symat3loops} to obtain manifest SUSY in $\mathcal{N}=2$ SYM at three loops.

\subsubsection*{Absence of $\gamma_5$ up to four loops (inclusive)}

Models with two real scalars, such as NJLY and the four‑dimensional WZ model Lagrangian, when written in terms of Dirac spinors, contain a $\gamma_5$ at the Yukawa vertex and therefore necessarily encounter the $\gamma_5$ problem. 
However, as noted in eq.~\eqref{eq:neq4}, the Levi–Civita tensor (and hence $\gamma_5$) contributes only in traces containing at least four $\Gamma^I$ matrices. 
Moreover, for the Levi–Civita terms to yield a non‑vanishing contribution to the kind of single momentum massless integrals we have, they must be contracted with one another. 
This requires a diagram with at least two fermion chains, each containing a trace with at least four matrices. 
Such a configuration appears for the first time in four loop propagators, in diagrams of type shown in Figure \ref{fig:urdhpund}, but the integrals evaluate to zero. The reason for this can be understood by considering the Passarino-Veltman reduction \cite{Passarino:1978jh} of the fermion loop subgraphs. For the Levi-Civita tensor $\epsilon^{\mu_1\mu_2\mu_3\mu_4}$ to survive, we need a tensor structure that has four distinct momenta, but as the subgraphs are four point functions and thus have only three independent momenta, that is not possible.

The first non-zero contribution can occur for the first time at five loops, with the rightmost diagram in Figure \ref{fig:urdhpund}. As now the fermion loop subgraph is a six-point function, with one mometum nullified due to infrared (IR) rearrangement, it has four independent momenta, and thus the Levi-Civita term survives. 
For computing the renormalisation group coefficients of the WZ model at five loops, we can still avoid the $\gamma_5$, as we only require to calculate the two-point five-loop functions which do not have a $\gamma_5$ and can then use eq.~\eqref{susyanomb} to compute the $\beta$-function~\footnote{
Reference~\cite{Kamenshchik:1983ufs} studies the WZ model in dimensional regularization using a formulation with chiral superfields. It reports a breakdown of SUSY Ward identities starting at five loops. We do not see any indication of such a problem in our approach, but a detailed analysis is left for future work.}.
However, as we are discussing this for the NJLY model in general, an explicit prescription for $\gamma_5$ from five loops onwards is required.
\begin{figure}
    \centering
     \flescapeo  \flescapet  \raisebox{-1cm}{\includegraphics[width=0.15\linewidth]{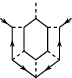}}
    \caption{Representative diagrams where non-zero $\gamma_5$ integrands appear for the first time in the NJLY and WZ models (all possible permutations of connecting the internal scalar lines (dashed lines, see footnote~\ref{ft:scalars}) to Yukawa vertices are considered in every diagram above). 
    The four‑loop diagrams evaluate to zero after performing the integration.
    Only the five-loop diagram on the right can have a non-zero $\gamma_5$ contribution. Thus, we can safely ignore $\gamma_5$ through four loops.}
    \label{fig:urdhpund}
\end{figure}
Hence, up to four loops we can disregard the $\gamma_5$ matrices entirely, and the resulting calculations remain fully accurate.

As a side remark, we note that this key mechanism is not confined to WZ models alone, but also appears in SYM theories.
A long-standing problem in SUSY regularisation was the observation that, SUSY was preserved in the results of $\mathcal{N}=1,4$ SYM, but not for $\mathcal{N}=2$ SYM results at three loops ~\cite{Avdeev:1982xy,Velizhanin:2008rw}. 
This was resolved in ~\cite{n2symat3loops}, which revealed that the apparent breakdown of SUSY was caused by this same $\gamma_5$-related flavour Clifford‑algebra subtlety, which unlocked at three loops for $\mathcal{N}=2$ SYM, but not for $\mathcal{N}=1,4$ SYM. Correcting that issue restored the expected equality of the Yukawa and gauge $\beta$-functions, showing that SUSY is in fact not violated.

The same mechanism implies that $\mathcal{N}=4$ SYM would first be sensitive to such effects only at five loops, consistent with existing four‑loop results, confirming that dimensional reduction preserves SUSY through all known loop orders.

\section{Diagrammatic Computations}
\label{sec:computations}

\subsection{Renormalising $\lgen$}\label{sec:renlgen}
In this section, we follow the standard renormalisation procedure for quantum field theories, extending it in a straightforward way to the case of two couplings. 
First, by locality of counterterms, the renormalisation scale $\mu$ can enter the counterterms only through the renormalised couplings $a(\mu)$ and $\lambda(\mu)$. 
Thus, we have
\begin{equation}\label{gammas}
    \gamma_i = \frac{d\ln Z_i}{d\ln\mu} = \beta_a\frac{\partial\ln Z_i}{\partial a}+\beta_\lambda\frac{\partial\ln Z_i}{\partial \lambda}\;,
\end{equation}
where the subscript $i=\phi,\psi$ depends on the field anomalous dimension under consideration. 
Next, for the $\beta$ functions, we take the derivative of both sides of eq.~\eqref{couplings} with respect to $\mu$, and obtain,
\begin{equation*}
    \begin{aligned}
        -2\eps(\lambda Z_{\lambda\lambda}+aZ_{a\lambda}) &= \left(Z_{\lambda\lambda}+\lambda\frac{\partial Z_{\lambda\lambda}}{\partial \lambda}\right)\beta_\lambda+\left(Z_{a\lambda}+a\frac{\partial Z_{a\lambda}}{\partial a}+\lambda\frac{\partial Z_{\lambda\lambda}}{\partial a}\right)\beta_a\, ,\\
        -2\eps(\lambda Z_{\lambda a}+aZ_{aa}) &= \left(Z_{\lambda a}+\lambda\frac{\partial Z_{\lambda a}}{\partial \lambda}+a\frac{\partial Z_{a a}}{\partial \lambda}\right)\beta_\lambda+\left(Z_{aa}+a\frac{\partial Z_{aa}}{\partial a}\right)\beta_a\, .
    \end{aligned}
\end{equation*}
We solve these equations for the $\beta$ functions in terms of the $Z$-factors. 
For notational simplicity, consider the following shorthand for the above set of equations:
\begin{equation*}
    \begin{bmatrix}
        c_\lambda \\
        c_a
    \end{bmatrix} =
    \begin{bmatrix}
        f_{\lambda\lambda} & f_{a\lambda} \\
        f_{\lambda a} & f_{aa}
    \end{bmatrix}
    \begin{bmatrix}
        \beta_\lambda \\
        \beta_a
    \end{bmatrix}\, .
\end{equation*}
Then the $\beta$ functions can be expressed as
\begin{equation}\label{betas}
    \beta_\lambda = \frac{1}{\det F}(f_{aa}c_\lambda - f_{a\lambda}c_a)\quad , \quad \beta_a = \frac{1}{\det F}(f_{\lambda\lambda}c_a - f_{\lambda a}c_\lambda)
    \, ,
\end{equation}
where $\det F = f_{aa}f_{\lambda\lambda}-f_{a\lambda}f_{\lambda a}$. 
In this form, eqs.~\eqref{gammas} and \eqref{betas} can be used to obtain the $Z$-factors in terms of the anomalous dimensions and $\beta$ functions order by order in perturbation theory. 
With the formal perturbative expansion, 
\begin{align*}
    \gamma_i &= \sum_{L=1}^\infty\sum_{k=0}^L\gamma_i^{(L,k)}\lambda^ka^{L-k}\, ,\\
    \beta_\lambda & = -2\eps\lambda  + \sum_{L=1}^\infty\sum_{k=0}^{L+1} \beta_\lambda^{(L,k)}\lambda^ka^{L+1-k}\, ,\\
    \beta_a & = -2\eps a  + \sum_{L=1}^\infty\sum_{k=0}^{L+1} \beta_a^{(L,k)}\lambda^ka^{L+1-k}\, ,
\end{align*}
the counterterms can be organised systematically as a series in the renormalised couplings. 
In the $\overline{\text{MS}}$-scheme, the $Z$-factors for any process admit a solution in the form of a Laurent series in $\eps$, given by 
\begin{align*}
    Z_i = & \;1 - \eps^{-1}(\lambda\gamma_i^{(1,1)}+ a\gamma_i^{(1,0)})\\
    &+\eps^{-2}\lambda^2\left(\frac{1}{2}(\gamma_i^{(1,1)})^2-\frac{1}{2}\beta_\lambda^{(1,2)}\gamma_i^{(1,1)}-\frac{1}{2}\beta_a^{(1,2)}\gamma_i^{(1,0)}\right)\\
    &+\eps^{-2}\lambda a\left(\frac{1}{2}\gamma_i^{(1,1)}\gamma_i^{(1,0)}-\frac{1}{2}\beta_\lambda^{(1,1)}\gamma_i^{(1,1)}-\frac{1}{2}\beta_a^{(1,1)}\gamma_i^{(1,0)}\right)\\
    &+\eps^{-2}a^2\left(\frac{1}{2}(\gamma_i^{(1,0)})^2-\frac{1}{2}\beta_\lambda^{(1,0)}\gamma_i^{(1,1)}-\frac{1}{2}\beta_a^{(1,0)}\gamma_i^{(1,0)}\right)\\
    &-  \eps^{-1}\left(\frac{1}{2}\lambda^2\gamma_i^{(2,2)}+ \frac{1}{2}a\lambda\gamma_i^{(2,1)}+\frac{1}{2}a^2\gamma_i^{(2,0)}\right) + \ldots
\end{align*}
Here, the term ``process" denotes any 1PI correlator, for instance, a propagator-related two-point function or a vertex-related three-point function. 
If its bare 1PI form is  
\begin{equation}
   \Gamma_i^{\text{bare,1PI}} = 1+\mathcal{O}(\lambda_0,a_0)\, ,
\end{equation}
with the tree‑level contribution normalised to one by suitable projections, then, using the $Z$-factor, the renormalised correlator is obtained via  
\begin{equation}
    \Gamma_i^{\text{ren,1PI}} = Z_i\Gamma_i^{\text{bare,1PI}} \;.
\end{equation}
To obtain the counterterms, we require that the pole part on the right‑hand side be set to zero and solve the resulting equations order by order in the couplings. The cancellation of higher-order poles by lower‑loop information serves as a non‑trivial consistency check of the integral computations. 

This procedure, which determines the $Z$-factors from the anomalous dimensions and $\beta$ functions, is useful in two main ways:
\begin{enumerate}
    \item The expression we obtain for the counterterms is independent of the details of the theory; it applies to any two‑coupling model renormalised in the $\overline{\text{MS}}$ scheme. It therefore provides a robust check on the calculations through the requirement of consistent divergence cancellation.
    
    \item By renormalising the entire Green's function in a single step, the method avoids the diagram‑by‑diagram combinatorics of the $R$‑operation and is therefore highly efficient. The full four‑loop renormalisation can be completed in a matter of seconds.
\end{enumerate}

\subsection{Methodology for calculations} \label{sec:methodology}

For the full computation of the theories defined by $\lgen$, we generate all diagrams up to four loops for propagators and vertices using \texttt{Qgraf}~\cite{Nogueira:1991ex}. 
The output is then processed with the computer‑algebra system \texttt{FORM}~\cite{Vermaseren:2000nd,Kuipers:2012rf,Ruijl:2017dtg,Davies:2026cci}, which we use to implement the Feynman rules, identify topologies, and perform the colour, flavour, and spinor traces, including the necessary Levi–Civita tensors. 
\texttt{FORM} also converts the \texttt{Qgraf} output into a format suitable for the database system \texttt{Minos}~\cite{minos}, which organises and evaluates the diagrams systematically.

\texttt{Minos} handles the bookkeeping of topologies and flavour factors and provides access to the individual diagram contributions. 
For the loop integrations, we interface \texttt{Minos} with \texttt{FORCER}~\cite{Ruijl:2017cxj}, which evaluates the required two‑point integrals up to four loops. 
The assignment of \texttt{FORCER} integrals to their corresponding flavour factors and the assembly of the full expressions for propagators and vertices are managed automatically by \texttt{Minos}. 
This entire chain—from diagram generation to integral evaluation—runs efficiently, with even the highest‑loop contributions processed within seconds. As a result, the computational overhead arising from the flavour structure of the theory remains minimal. 
For renormalisation, we employ the ansatz described in Sec.~\ref{sec:renlgen} and implemented in an automated \texttt{FORM} code. 
The optimised procedure performs the full four‑loop renormalisation of the theory within seconds.

The most computationally expensive part of our calculations is the use of \texttt{FORCER}, which evaluates integrals through integration-by-parts identities. 
\texttt{FORCER} is specifically designed for two‑point massless diagrams up to four loops, making it significantly more efficient than general integral reduction algorithms. This is ideally suited for renormalisation, since the counterterms of all relevant propagators and vertices can typically be expressed entirely in terms of two‑point massless diagrams. For vertex functions, where the UV divergence is only logarithmic, we apply IR rearrangement: one can always choose one external momentum for the three‑point vertex, or two external momenta for the four‑point vertex, to be set to zero without affecting the counterterms.

In the present work, we compute the four‑loop counterterms for the fields and for the Yukawa vertex. Owing to technical difficulties in extending our automated framework to implement IR rearrangement for the quartic vertex in an IR‑safe manner, the four‑loop renormalisation of the quartic vertex is not included.

\subsection{Gross-Neveu-Yukawa model from $\lgen$}

We first demonstrate that the results for the GNY model can be recovered from those of $\lgen$ by simply setting $n_s=1$. 
For comparison, we use the three‑loop results of \cite{Mihaila:2017ble} and the four‑loop results of \cite{Zerf:2017zqi}. The conventions employed in these references differ from those adopted in our calculation, specifically:
\begin{equation}\label{eq:maptozerf}
\lambda=4\lambda\Big|_{\text{\cite{Zerf:2017zqi,Mihaila:2017ble}}}\quad,\quad a= \frac{1}{2}a\Big|_{\text{\cite{Zerf:2017zqi,Mihaila:2017ble}}}\quad,\quad \eps=\frac{1}{2}\eps\Big|_{\text{\cite{Zerf:2017zqi,Mihaila:2017ble}}}\quad,\quad n_f=2n_f\Big|_{\text{\cite{Zerf:2017zqi,Mihaila:2017ble}}}\;,
\end{equation}
where $|_{\text{\cite{Zerf:2017zqi,Mihaila:2017ble}}}$ denotes that the corresponding quantity is based on the conventions used in \cite{Zerf:2017zqi,Mihaila:2017ble}. We present the example of the $\beta$-function of the Yukawa vertex for $\lgen$ here. 
After matching the conventions, the three-loop result is given by:
\begin{equation}
\begin{aligned}
\beta^{\text{gen}}_a &= -\eps a+ a^{2}\!\left(4 - n_{s} + 2 n_{f}\right)
+ a\,\lambda^{2}\!\left(16 + 8 n_{s}\right)
+ a^{2}\lambda\!\left(-16 - 8 n_{s}\right) \\
&\quad
+ a^{3}\!\left(-8 + 8 n_{s} - \frac{9}{8} n_{s}^{2} - 6 n_{f}\right) \\
&\quad
+ a\,\lambda^{3}\!\left(-128 - 80 n_{s} - 8 n_{s}^{2}\right)
+ a^{2}\lambda^{2}\!\left(184 + 90 n_{s} - n_{s}^{2} - 60 n_{f} - 30 n_{f} n_{s}\right) \\
&\quad
+ a^{3}\lambda\!\left(96 + 36 n_{s} - 6 n_{s}^{2} + 60 n_{f} + 30 n_{f} n_{s}\right)\\
&\quad
+ a^{4}\!\left(20 - 41 n_{s} + \frac{189}{16} n_{s}^{2} - \frac{109}{64} n_{s}^{3} + 5 n_{f} - \frac{17}{4} n_{f} n_{s} + \frac{43}{32} n_{f} n_{s}^{2} + \frac{1}{2} n_{f}^{2} + 3 n_{f}^{2} n_{s}\right) \\
&\quad
+ a^{4}\zeta_{3}\!\left(36 - 27 n_{s} + \frac{9}{2} n_{s}^{2} + \frac{3}{4} n_{s}^{3} + 18 n_{f} - 3 n_{f} n_{s} - \frac{3}{2} n_{f} n_{s}^{2}\right)
\, .
\end{aligned}
\end{equation}
By setting $n_s = 1$ in the $\lgen$ result, we get the result for the GNY model:
\begin{align}
\beta_{a}^{\text{GNY}} =&\ -\eps a + (3 + 2 n_{f})\, a^{2} + 24\, a \lambda^2 - 24\, a^2 \lambda -  a^{3}\big(\frac{9}{8} + 6 n_{f}\big) - 216 a\lambda^{3} \nonumber\\
&\quad +a^2 \lambda^{2} (273 - 90 n_{f}) + a^{3} \lambda(126+90 n_{f}) \nonumber \\
&\quad + a^{4}\left(-\frac{697}{64} + \frac{67}{32} n_{f} + \frac{7}{2} n_{f}^{2}\right)+ \frac{a^{4}}{4} \zeta_{3}\!\left(57 + 54 n_{f}\right) \, ,
\end{align}
while the four-loop contribution is,
\begin{equation}
\begin{aligned}
\beta_a^{\text{GNY},(4)}&= 14040\, a \lambda^{4}
+ a^{2}\lambda^{3}\!\left(-16380 + 288\,n_{f}\right)
+ 5184\, a^{2}\lambda^{3}\zeta_{3} \\
&\quad
+ a^{3}\lambda^{2}\!\left(-\frac{4455}{2} - 1270\,n_{f} - 12\,n_{f}^{2}\right) 
+ a^{3}\lambda^{2}\zeta_{3}\!\left(-2700 + 648\,n_{f}\right)\\
&\quad
+ a^{4}\lambda\!\left(-\frac{2829}{8} - 683\,n_{f} + 12\,n_{f}^{2}\right)
+ a^{4}\lambda\zeta_{3}\!\left(-378 - 648\,n_{f}\right) \\
&\quad
+ a^{5}\zeta_{4}\!\left(\frac{171}{8} + \frac{69}{2}\,n_{f} + \frac{27}{2}\,n_{f}^{2}\right)+ a^{5}\zeta_{5}\!\left(-\frac{215}{2} - 105\,n_{f}\right) \\
&\quad
+ a^{5}\zeta_{3}\!\left(\frac{5}{8} - \frac{331}{2}\,n_{f} - \frac{125}{2}\,n_{f}^{2}\right) + a^{5}\Delta_{3}n_f(1+107\zeta_3-125\zeta_5)
\\
&\quad +  a^{5}\!\left(\frac{30529}{512} + \frac{9907}{64}\,n_{f} - \frac{899}{24}\,n_{f}^{2} + \frac{11}{6}\,n_{f}^{3}\right)
\, ,
\end{aligned}
\end{equation}
matching the results of \cite{Zerf:2017zqi,Mihaila:2017ble}. We note that the term $\Delta_3$ represents the additional three‑dimensional contribution arising from the odd traces of the diagrams of the type shown in figure \ref{fig:oddtrace}. 
In three dimensions, one has $\Delta_3=1$, while in four dimensions $\Delta_3=0$.  
This illustrates how the generalised Lagrangian is capable of reproducing the correct results for the GNY model. Likewise, the results for the NJLY model are obtained directly by substituting $n_s=2$ into the expressions derived from $\lgen$, without the need to perform a separate calculation for that model.

\subsection{Emergent SUSY in \texorpdfstring{$\lgen$}{lgen}}

We now investigate whether there exists a \emph{supersymmetric point} $(n_s,n_f)$ within this Lagrangian. 
We refer to a point as supersymmetric if the SUSY Ward identities are satisfied for those specific values of $n_s$ and $n_f$. To demonstrate the emergence of SUSY in $\lgen$, we employ the SUSY Ward identity relating the anomalous dimensions of the fermion and scalar fields. 
Accordingly, we search for all-order solutions in perturbation theory of the equation
\begin{equation}
\label{eq:gf=gs}
  \gamma_f = \gamma_s\, ,  
\end{equation}
where $\gamma_f$ and $\gamma_s$ denote the anomalous dimensions of the fermionic and scalar fields, respectively. 
Since each order in the coupling constant(s), as well as each transcendental structure (such as factors of 
$\zeta(n)$, is linearly independent, the equality must yield a distinct set of equations at every loop order and for every transcendental coefficient. 
The condition in eq.~\eqref{eq:gf=gs} is therefore satisfied if and only if all these equations possess a non‑empty set of simultaneous solutions.

In the following, we present the curves corresponding to each loop order up to $L=4$, for which the $L$-loop anomalous dimensions of the fermion and scalar fields coincide. 
A \emph{supersymmetric point}, if it exists, is identified as a point in parameter space at which all of these curves intersect simultaneously.
\begin{itemize}
    \item tree level: $\lambda = g^2$ 

    \item one-loop: $\displaystyle n_s - 2 n_f = 0$
    \hfill (equality of bosonic and fermionic degrees of freedom)

    \item two-loop: 
    \[
        -\frac{1}{8} n_s^2 - \frac{1}{4} n_f n_s - \frac{1}{4} n_s
        + 2 n_f - \frac{1}{2} = 0
    \]

    \item three-loop:
    \[
        -\frac{17}{8} n_f^2 n_s - \frac{9}{2} n_f^2
        - \frac{13}{8} n_f n_s^2
        + \frac{283}{16} n_f n_s
        - \frac{153}{8} n_f
        + \frac{41}{32} n_s^3
        - \frac{225}{32} n_s^2
        + \frac{127}{16} n_s
        + 1 = 0
    \]

    \item three-loop $\zeta(3)$:
    \[
        (n_s - 2 n_f)\bigl(n_s^2 - 6 n_s + 4\bigr) = 0
    \]

    \item four-loop $\zeta(4)$:
    \[
        (n_s - 2 n_f)\Bigl(
            n_f n_s^2 - 10 n_f n_s + 12 n_f
            - n_s^3 + 18 n_s^2 - 60 n_s + 48
        \Bigr) = 0
    \]

    \item four-loop $\zeta(5)$:
    \[
        (n_s - 2 n_f)\bigl(
            n_s^3 - 8 n_s^2 + 24 n_s - 16
        \bigr) = 0
    \]

    \item four-loop $\zeta(3)$:
    \begin{align*}
        0 ={}&
        - n_f^3 n_s + 8 n_f^3
        - n_f^2 n_s^2 + 10 n_f^2 n_s + 44 n_f^2
        - \frac{1}{4} n_f n_s^3 - \frac{49}{2} n_f n_s^2
        + 63 n_f n_s - 68 n_f \\
        &\quad
        + 2 n_s^4 - \frac{7}{4} n_s^3
        - \frac{43}{2} n_s^2 + 22 n_s
    \end{align*}

    \item four-loop:
    \begin{align*}
        0 ={}&
        \frac{181}{48} n_f^3 n_s - \frac{49}{6} n_f^3
        + n_f^2 n_s^2 + \frac{217}{8} n_f^2 n_s
        - \frac{389}{12} n_f^2
        - \frac{1157}{96} n_f n_s^3 + \frac{67}{2} n_f n_s^2 \\
        &\quad
        - 129 n_f n_s + 195 n_f
        + \frac{545}{128} n_s^4
        - \frac{3773}{192} n_s^3
        + \frac{1525}{24} n_s^2
        - \frac{597}{8} n_s
        - \frac{125}{8}
    \end{align*}
\end{itemize}
All of these curves intersect at two points: $(2,1)$ and $\left(1, \tfrac{1}{2}\right)$. 
To illustrate this, in Fig.~\ref{fig:emgsusytheo} we display the curves corresponding to the lowest transcendental weights, those yielding the most intricate structures, at each loop order from one to four.
\begin{figure}
    \centering
    \includegraphics[width=0.5\linewidth]{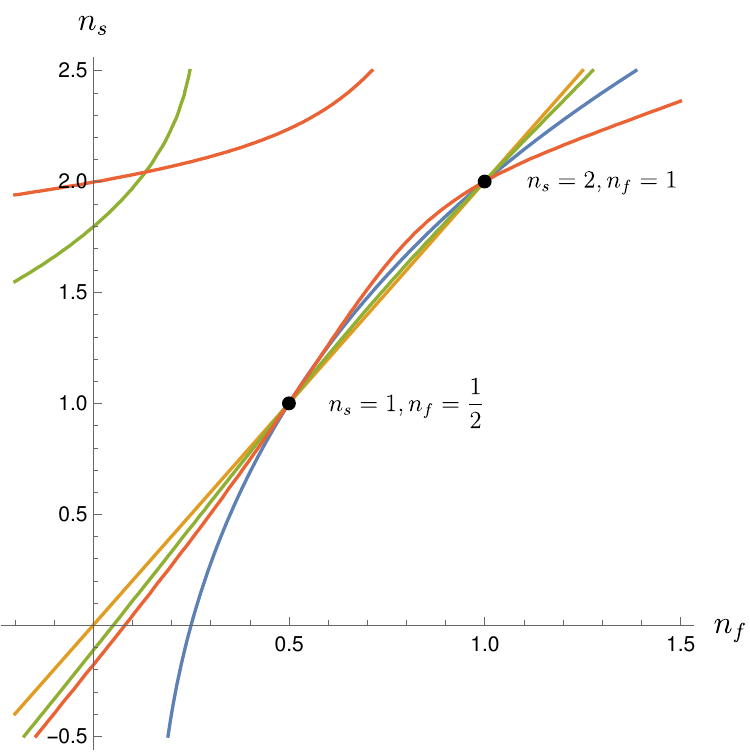}
    \caption{Plot of the SUSY Ward Identity $\gamma_f=\gamma_s$ at every loop order up to 4 loops, y-axis is $n_s$, x-axis is $n_f$. Curves for every loop order intersect at the same 2 points, signifying that $\lgen$ as a whole becomes supersymmetric at those two points.}
    \label{fig:emgsusytheo}
\end{figure}

Since we employ Weyl fermions, the value $n_f = 1$ corresponds to two fermionic degrees of freedom. Consequently, the first intersection point represents a theory with two scalar and two fermionic degrees of freedom, precisely matching the field content of the $\mathcal{N}=1$ chiral multiplet. 
In \cite{Zerf:2017zqi}, Dirac spinors were used instead; in that convention, SUSY emerges at $n_f = 1/2$, which corresponds to the same number of fermionic degrees of freedom. 

The point $\left(1, \tfrac{1}{2}\right)$ corresponds to a theory with a single fermionic degree of freedom. In four dimensions, such a field content could only arise from a Weyl–Majorana fermion; however, no irreducible Weyl–Majorana representation exists in four dimensions. Interestingly, this particle content coincides with that of the three‑dimensional WZ model~\cite{Grover:2013rc}, which contains one real scalar and one two‑component Majorana fermion, the latter possessing one on‑shell fermionic degree of freedom.

To give further evidence that those are indeed the emergent SUSY points for the entire generalised Lagrangian, we can plot the curves for the non-renormalisation of the superpotential (see Figure \ref{fig:non-renorm}):
\begin{equation}
    \beta_a = -3\gamma_{\phi}\,.
\end{equation}
The above relation only emerges at the point $(2,1)$, but not at the point $\left(1, \tfrac{1}{2}\right)$. 
This is precisely consistent with the fact the former point corresponds to the four-dimensional WZ model where holomorphy exists and non-renormalisation theorems hold, whereas the latter emergent SUSY point has no holomorphy as it corresponds to an $\mathcal{N}=1$ three-dimensional SUSY.
\begin{figure}
    \centering
    \includegraphics[width=0.5\linewidth]{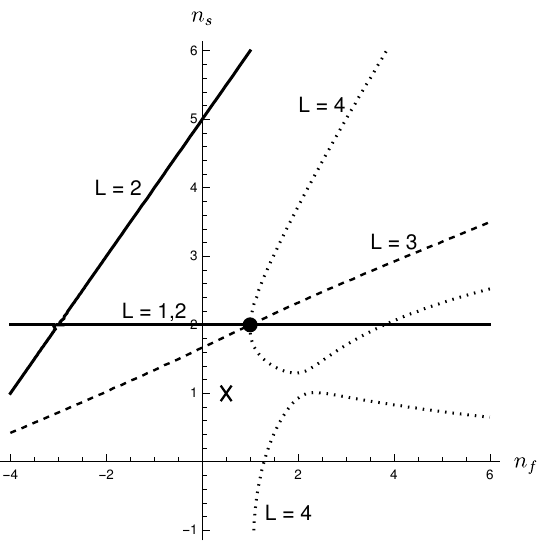}
    \caption{Curves for the non-renormalisation of the super-potential, $\beta_a = -3\gamma_{\phi}$ at every loop order up to four loops. All the curves intersect at the four-dimensional WZ point $(2,1)$, denoted by the black dot, but not at the three-dimensional WZ point $\left(1, \tfrac{1}{2}\right)$, denoted by the black cross. This is consistent because the four-dimensional superpotential is holomorphic but the three-dimensional one is not.}
    \label{fig:non-renorm}
\end{figure}

Quantitatively, the renormalisation constants at the two emergent SUSY points agree with existing four-loop calculations made directly in the four-dimensional WZ model,~\cite{Avdeev:1982jx} after making the substitutions in eq.~\eqref{eq:maptozerf}:
\begin{equation}
    \beta_{a}^{\,\text{4D-WZ}} = 3a^2 - 3a^3 + a^4\left(\frac{15}{4} + 9\zeta_3\right) + a^5\left(-\frac{27}{4} + \frac{27}{2}\zeta_4 - 60\zeta_5 - 45\zeta_3\right)
\, ,
\end{equation}
confirming that this intersection indeed corresponds to the four-dimensional WZ model. 

As mentioned in Sec.~\ref{sec:methodology}, the direct computation of the quartic vertex $\beta$ function for $\lgen$ has not been carried out in this work. However, our four-loop computation of the Yukawa vertex $\beta_a$ function is sufficient to establish that the following Ward identity also holds through four loops at the same two emergent SUSY points:
\begin{equation}
    \beta_a = \beta_{\lambda},
\end{equation}
By mapping our result for the $\beta_a$ function in $\lgen$ to the results in~\cite{Zerf:2017zqi} using eq.~\eqref{eq:maptozerf}, we can see that at the two emergent SUSY points  $\left(1, \tfrac{1}{2}\right)$ and  $(2,1)$ from $\lgen$ become equal to $\beta_\lambda$ for the GNY and NJLY models at their respective SUSY points. 
The check for this is included in the ancillary file. 
This quantitative analysis shows that these two cases are the only ones in which emergent SUSY can be realized. Starting from the most general form of a renormalisable theory with scalars and fermions, and imposing only flavour non-evanescence, we find that these two possibilities exhaust all renormalisable supersymmetric theories.

We note a key distinction between the emergent SUSY observed in $\lgen$ and the previously known cases discussed in~\cite{Mihaila:2017ble,Zerf:2017zqi,Grover:2013rc}. 
In the existing literature, the two supersymmetric points arise from two distinct theories, the GNY and NJLY models, without any underlying unifying framework. 
In contrast, within $\lgen$ both supersymmetric points, as well as the GNY and NJLY models themselves, emerge naturally. 
In this sense, $\lgen$ provides a unified description of all these theories, and gives twice the amount of Ward Identities from the emergent SUSY points.

In the next section, we demonstrate how emergent SUSY can be exploited to optimise perturbative calculations. As we will see, $\lgen$ supplies additional SUSY Ward identities that remain applicable even when studying a non‑supersymmetric theory, thereby enabling the simplification of computations by reducing the number of integrals to be computed.

\subsection{Optimisation with Emergent SUSY}\label{sec:application}

The method of employing emergent SUSY to optimise the calculations proceeds as follows:
\begin{enumerate}
\item  Express the 1PI correlator in its most general form, decomposed into flavour-dependent tensor structures multiplied by flavour‑independent loop integrals.

\item  Insert the specific flavour assignments and coupling values for which emergent SUSY is conjectured to arise from $\lgen$. This substitution yields a set of relations among the flavour‑independent integrals.

\item Employ these relations to eliminate the most technically demanding integrals from the general expressions of the 1PI correlators, thereby minimising the number of loop integrations that must be evaluated explicitly.

\item Finally, substitute the flavour values corresponding to the physical theory for which the correlators are to be computed.
\end{enumerate}

\noindent 
Take, for example, the two‑loop propagator 1PI correlators for the fermion (solid lines) and the scalar fields (dashed lines).  
In step 1, we obtain the following general relation:
\begin{equation}\label{eq:emgsusytwoloop}
    \begin{aligned}
        \Gamma_\psi^{(2)} &= a_0^2n_s(n_s-2)\,\intpsio +a_0^2n_fn_s\,\intpsitw \, ,\\
        \Gamma_\phi^{(2)} &= a_0^2n_f(n_s-2)\,\intphio+a_0^2n_fn_s\,\intphitw + \lambda_0^2(n_s+2)\,\intphithr\;.
    \end{aligned}
\end{equation}
Note that the diagrams displayed above are understood to be evaluated with the appropriate Born projections, such that the corresponding tree‑level contributions are normalised to unity.  

Proceeding to step 2, we impose the condition $a_0=\lambda_0$ under which emergent SUSY is conjectured to arise. 
At the SUSY points, the supersymmetric Ward identity  
$\Gamma_\psi^{(2)}=\Gamma_\phi^{(2)}$ must hold.
Substituting for $(n_s, n_f)$ the pair of values $\left(1, \tfrac{1}{2}\right)$ and 
 $\left(2, 1\right)$ yields linear relations among the corresponding integrals,
\begin{equation}
\label{eq:optimizing-susy-1}
    \begin{aligned}
         -1\,\intpsio + \frac{1}{2}\intpsitw&= -\frac{1}{2}\intphio + \frac{1}{2}\intphitw+3\intphithr\, ,\\
        2\intpsitw &=2\intphitw+4\intphithr \; .
    \end{aligned}
\end{equation}
Now proceed to step 3. The two relations 
in eq.~\eqref{eq:optimizing-susy-1} allow us to express two of the integrals in terms of the remaining ones:
\begin{equation}
    \begin{aligned}
    \intphithr &=\frac{1}{2}\intpsitw - \frac{1}{2}\intphitw \;,\\
         \intpsio &= \frac{1}{2}\intphio + \frac{1}{2}\intphitw -\frac{1}{2}\intpsitw\, .
    \end{aligned}
\end{equation}
Substituting the above expressions into eq.~\eqref{eq:emgsusytwoloop} completes step 3. Consequently, the number of independent integrals is reduced from five to three. To apply the optimised expressions to the GNY model, one then sets $n_s=1$, which constitutes step 4.  

It should be emphasised that the integrals appearing in eq.~\eqref{eq:emgsusytwoloop} are not scalar integrals; they still carry their full spinor structures, with only the flavour decomposition having been performed. 
The reduction in the number of independent integrals therefore directly decreases the number of objects that must be subjected to the standard integration-by-parts reduction, accelerating the computationally most demanding part of the analysis.

At first sight, it appears that at any given loop order we can eliminate only two integrals, since the two SUSY‑point Ward identities provide only two relations. This would suggest that the optimisation procedure becomes ineffective at higher loop orders, where the number of independent integrals grows faster than factorially.  

However, at higher loops the situation improves: the most intricate topologies often satisfy the SUSY Ward identities independently of the remaining ones, thereby yielding additional, topology‑specific relations. These extra equations allow further eliminations beyond the two obtained from the SUSY points alone.  

For example, consider the three‑loop non‑planar topologies:
\begin{equation}
    \nonpphi\, ,\qquad \nonppsi\, .
\end{equation}
Again, the diagrams displayed above are understood to be evaluated with the appropriate Born projections. In this case, the two non‑planar three‑loop diagrams satisfy a SUSY Ward identity independently of all other topologies. Explicitly, they obey the relation~\footnote{This relation can also be seen with an adaptation of the cut-and-glue technique of~\cite{Baikov:2010hf}. Closing the open ends of each two-point function with a propagator leads to the same topology up to a symmmetry factor.}:
\begin{equation}
    \nonpphi=\quad2\; \nonppsi
\end{equation}
The origin of this behaviour may possibly be related to the independence of the corresponding master integrals, a point that could be investigated in a dedicated future study. 

An explicit application of the optimisation procedure to the case of anomalous dimensions of operator matrix elements (OMEs) will be presented in the next section.

\section{Application: Optimising operator renormalisation}
\label{sec:omes}

Anomalous dimensions of twist‑two operators determine the scale dependence of parton distribution functions in QCD and depend on the spin $N$. At low loop order one can compute closed expressions in $N$ for the relevant OMEs, but at higher loops the growing complexity of the topologies typically forces one to compute many fixed‑$N$ values and reconstruct the general $N$ dependence afterward, see, e.g.~\cite{Falcioni:2025hfz}.
As the loop order increases, both the complexity of the general $N$ result and the cost of computing each fixed‑$N$ value grow sharply, often requiring weeks or months of computing time. 
The use of SUSY can reduce the number of required integrals, see for example~\cite{Moch:2018wjh}.

In the following sections, we show how the generalised Lagrangian and the optimisation technique introduced in Sec.~\ref{sec:application} can systematically use SUSY to streamline operator anomalous‑dimension calculations. We develop the tool here using the GNY model as a proof-of-concept, to be applyied to theories like QCD in the future.

Due to the $SU(n_f)$ flavour symmetry of the fermions in the GNY model, there is one non‑singlet fermion operator
\begin{equation}
\label{eq:opns-def}
 \mathcal{O}^{\mu_1\mu_2\ldots\mu_N;\;i}_\psi(N) = \frac{(-i)^{N-1}}{N!}\bar\psi_a\gamma^{[\mu_1}\partial^{\mu_2\ldots\mu_N]}t^i_{ab}\psi_b\:-\:\text{traces}
 \, .
\end{equation}
The following notation and conventions are used in eq.~\eqref{eq:opns-def}:
\begin{itemize}
    \item The Lorentz indices of the operator $\mathcal{O}^{i}_\psi$ have been supressed for brevity, but it carries an $N$-index tensor structure.

    \item As $\mathcal{O}^{i}_\psi$ transforms in an irreducible spin-$N$ representation of Lorentz algebra, it is traceless in the Lorentz indices $\mu_i$.

    \item $t^i$ is a generator of the $\mathfrak{su}(n_f)$ Lie algebra. 

    \item $\partial^{\mu_i\ldots\mu_N}=\partial^{\mu_i}\partial^{\mu_{i+1}}\ldots\partial^{\mu_N}$, and $[\mu_1\ldots\mu_N]$ denotes all possible symmetric combinations of the Lorentz indices $\mu_i$.
\end{itemize}
In addition to this, we also have two flavour-singlet operators:
\begin{align}   
\mathcal{O}^{\mu_1\mu_2\ldots\mu_N}_\psi(N) &= \frac{(-i)^{N-1}}{N!}\bar\psi_a\gamma^{[\mu_1}\partial^{\mu_2\ldots\mu_N]}\psi_a\:-\:\text{traces}\\
    \mathcal{O}^{\mu_1\mu_2\ldots\mu_N}_\phi(N) &= \frac{(-i)^{N}}{N!}\phi\,\partial^{[\mu_1\ldots\mu_N]}\,\phi\:-\:\text{traces}
\end{align}
Contracting with $n_{\mu_1}n_{\mu_2}\ldots n_{\mu_N}$ where $n_\mu$ is a light-like vector, the trace parts vanish and we have the following light cone projected versions of the operators:
\begin{align}  
\mathcal{O}^{i}_\psi(N) &=\bar\psi_a\gamma_+(-i\partial_+)^{N-1}t^i_{ab}\psi_b\\
\mathcal{O}_\psi(N) &= \bar\psi_a\gamma_+(-i\partial_+)^{N-1}\psi_a\\
    \mathcal{O}_\phi(N) &= \phi\,(-i\partial_+)^{N}\,\phi
\end{align}
where the subscript $+$ comes from the notation of light-cone coordinates, wherein $p.n=p_+$. Note that the operator $\mathcal{O}_\phi(N)$ only exists for even $N$ as it is antisymmetric for odd $N$ modulo total derivatives, which vanish for the case of zero momentum flow through the operator considered here.

\subsection{Towards emergent SUSY} 

Our objective is to use $\lgen$ to map the operator Green's functions in the GNY model to those of the WZ model, from where we can use SUSY Ward identities to find relations among the integrals and hence eliminate some of them from the calculation.
Establishing how this optimisation operates in the GNY model provides a concrete template for the phenomenologically relevant case of QCD OMEs.

As the fermion flavour symmetry in the WZ model is only a $U(1)$ symmetry, it contains no non-trivial non-singlet operator, but only the following flavour singlet operators:
\begin{align}  
\mathcal{O}^{WZ}_\psi(N) &= \bar\psi\,\overline\sigma_+(-i\partial_+)^{N-1}\psi\, ,\\
    \mathcal{O}^{WZ}_\phi(N) &= \phi^\dagger\,(-i\partial_+)^{N}\,\phi
\, ,
\end{align}
where the $\psi$ is now a Weyl fermion and hence we have $\overline{\sigma}_+$ instead of $\gamma_+$. 

Thus, we can only expect to obtain a relation in the case of singlet operators in the GNY model, which is also the more involved case due to operator mixing. Note that both odd and even moments are defined for the WZ operators. In our generalised model, we will consider the following singlet operators:
\begin{align}  
\mathcal{O}^{\text{gen}}_\psi(N) &= \bar\psi_a\,\gamma_+(-i\partial_+)^{N-1}\psi_a\, ,\\
    \mathcal{O}^{\text{gen}}_\phi(N) &= \phi^I\,(-i\partial_+)^{N}\,\phi^I
\, .
\end{align}
We expect that the point $n_s=2,\,n_f=1/2$ yields the WZ Green’s functions, whereas $n_s=1$ reproduces those of the GNY model. Note that these generalised singlet operators are defined only for even moments. 
As a consequence, any emergent SUSY can occur only at even moments, since for odd moments the scalar operators of $\lgen$ and those of the WZ model do not coincide.

For the purpose of deriving the Ward identities, we promote the singlet operators to vertices of the generalised Lagrangian:
\begin{equation*}
 \mathcal{L}^{\text{gen}}_{\text{op},N} = \lgen+ p_0\mathcal{O}^{\text{gen}}_\phi(N)+q_0\mathcal{O}^{\text{gen}}_\psi(N)
 \, ,
\end{equation*}
where we have suppressed the bare subscripts for the fields. 
The renormalised operator couplings $p$ and $q$ are implicitly multiplied by the same Grassmann number $c$ so that 
\begin{equation}
    p^2=q^2=pq = 0\,.   
\end{equation}
This setup guarantees that each diagram in the theory contains at most one operator vertex, ensuring correspondence with the operator‑matrix‑element topology. 
Within this formalism, operator mixing is encoded in the mixing of the two coupling constants:
\begin{equation}\label{oper_ren}
    \begin{bmatrix}
        p_0 \\
        q_0
    \end{bmatrix} =\mu^{2-N}
    \begin{bmatrix}
        Z_{pp} & Z_{qp} \\
        Z_{pq} & Z_{qq}
    \end{bmatrix}
    \begin{bmatrix}
        p \\
        q
    \end{bmatrix}
\end{equation}
The $Z$-factors above are independent of $p$ and $q$. 
The $\beta$ functions of these couplings act as the anomalous dimensions of the operators.

\subsection{Operator Ward Identities in the WZ model}\label{sec:wardi}

To employ emergent SUSY as an optimisation tool, we must first derive a SUSY Ward identity for the operators. In this section, we summarise the key steps and results that lead to this Ward identity; the full derivation is provided in App.~\ref{app:wi}. 

The modified Lagrangian for the WZ model, including the operator vertices, is given by
\begin{equation}
 \mathcal{L}^{\text{WZ}}_{\text{op},N} = \mathcal{L}^{\text{WZ}}+\underbrace{ p_0\mathcal{O}^{\text{WZ}}_\phi(N)+q_0\mathcal{O}^{\text{WZ}}_\psi(N)}_{\mathcal{L}^{\text{WZ}}_{\mathcal{O}}}\, .   
\end{equation}
Then, the path integral with sources for this model is given by 
\begin{equation}
Z(J) = \int d\Phi\; e^{i(S+S_{\mathcal{O}}+\int d^Dx\, J^T\Phi)}\;,   
\end{equation}
where the actions are $S=\int d^dx\,\mathcal{L}^{\text{WZ}}$ and $S_{\mathcal{O}}=\int d^Dx\,\mathcal{L}^{\text{WZ}}_{\mathcal{O}}$.
The fields and source terms are abbreviated as 
$\Phi=(\phi,\phi^\dagger,\bar\psi_{\dot\alpha},\psi_\alpha)$ and $J =(J_\phi,J_{\phi^\dagger},-J_{\bar\psi}^{\dot{\alpha}},J_{\psi}^\alpha)$. 
From there, we can derive the following Ward identity (see App.~\ref{app:wi}):
\begin{equation}
    \langle\;\phi(y_1)\bar\psi_{\dot\alpha}(y_2)\delta_{\epsilon,\bar\epsilon} S_{\mathcal{O}}\;\rangle^{\small\text{WZ}} + \langle\;\phi(y_1)\delta_{\epsilon,\bar\epsilon}\bar\psi_{\dot\alpha}(y_2) S_{\mathcal{O}}\;\rangle^{\small\text{WZ}}+\langle\;\delta_{\epsilon,\bar\epsilon}\bar\phi(y_1)\psi_{\dot\alpha}(y_2) S_{\mathcal{O}}\;\rangle^{\small\text{WZ}} = 0\;,
\end{equation}
where the superscript WZ indicates that the correlator is based on the pure WZ action $S$ without the operator vertices. 
The infinitesimal SUSY transformations for the fields are given by:
\begin{align}
    \delta_{\epsilon,\bar\epsilon}\,\phi &= \epsilon^\alpha\psi_\alpha\, ,\\ 
    \delta_{\epsilon,\bar\epsilon}\,\bar\psi_{\dot\alpha} &= -i(\epsilon\sigma^\mu)_{\dot\alpha}\partial_\mu\phi^\dagger -\bar\epsilon_{\dot\alpha}F^\dagger\, .
    \label{fermionsusy}
\end{align}
Above, we have omitted overall factors of $\sqrt2$ on the right-hand sides, as they do not affect the Ward identity. However, the presence of the auxiliary field in the SUSY transformation laws introduces cumbersome terms in the variation $\delta_{\epsilon,\bar\epsilon} S_{\mathcal{O}}$ 
of the operator action, thereby rendering the resulting Ward identity ineffective for optimisation:
\begin{equation}\label{ward_explicit}
\begin{aligned}
    0  =&\:(-i\slashed\partial_{y_2})_{\alpha\dot\alpha}\langle\phi(y_1)\,S_\mathcal{O}\,\phi^\dagger(y_2)\rangle + \langle\psi_\alpha(y_1)\,S_\mathcal{O}\,\bar\psi_{\dot\alpha}(y_2)\rangle \\
    &+\langle\phi(y_1)\int d^Dx\left[p_0\mathcal{O}^1_{\alpha}(x)+q_0\mathcal{O}^2_{\alpha}(x)-q_0\mathcal{O}^F_{\alpha}(x)\right]\,\bar\psi_{\dot\alpha}(y_2)\rangle\;.
\end{aligned}
\end{equation}
Here we have new operators with open fermionic indices:

\begin{equation}
\begin{aligned}
    \mathcal{O}^1_{\alpha}& = \phi^\dagger(-i\partial_+)^N\psi_\alpha \\
    \mathcal{O}^2_{\alpha}& = \phi^\dagger(-i\partial_+)^{N-1}(-i\partial_\mu)(\sigma^\mu\bar{\sigma}_+\psi)_\alpha\\
    \mathcal{O}^F_{\alpha}& =(\sigma_+\bar\psi)_\alpha(-i\partial_+)^{N-1}F
\end{aligned}
\end{equation}
The two correlators in the first line of eq.~\eqref{ward_explicit} are precisely the quantities we aim to relate. The correlator in the second line, however, involves exotic operators, does not vanish, and introduces even more complicated structures. As a result, this equation fails to provide a useful emergent‑SUSY relation between the OMEs.

Fortunately, upon projecting to the light‑cone basis (appendix \ref{app:lightconesusy},) the auxiliary field is removed from the SUSY transformation \eqref{fermionsusy}. Such a method was used in \cite{Onishchenko:2003hs} for super-conformal operators. It enables us to obtain a meaningful relation between the OMEs: 
\begin{equation}\label{ward_ult}
\begin{aligned}
     0=&-i\sigma_-\partial_+\langle\phi(y_1)\,S_\mathcal{O}\,\phi^\dagger(y_2)\rangle + \langle\psi_-(y_1)\,S_\mathcal{O}\,\bar\psi_-(y_2)\rangle\\
    &+\langle\phi(y_1)\int d^Dx(p_0-2q_0)\mathcal{O}^1_-(x)\,\bar\psi_-(y_2)\rangle\;.
\end{aligned}
\end{equation}
The subscripts $+$ and $-$ denote light‑cone coordinates. Explicit definitions of the light‑cone elements are provided in App.~\ref{app:lightconesusy}. 
The Ward identity above is useful because, as shown in the following theorems, the correlator on the second line vanishes (see corollary~\ref{cor:peqq} discussed below), since $p_0=2q_0$. 
This yields a direct relation between the OMEs we seek to compute, which we will subsequently use for optimisation.

Using the Ward identity \eqref{ward_ult} -- whose formulation holds for all loop orders and all Mellin moments -- we obtain a relation among the operator anomalous dimensions that is analogous to those encountered in QCD. This phenomenon fits into a broader structural pattern: the SUSY combinations of anomalous dimensions, once transformed to the SUSY‑preserving dimensional‑reduction scheme, are known to vanish for both the unpolarised and polarised two‑loop splitting functions~\cite{Vogelsang:1996im}.

The all‑order SUSY‑motivated relations among the counterterms and anomalous dimensions of the singlet operators can be summarized as follows:
\begin{restatable}{theo}{mainthm}\label{the_theorem}
    The counterterms at any value of the moment $N$ are related to each other via the following equation in the $\overline{\text{MS}}$-scheme: 
    \begin{equation}\label{OME_ctm}
        2Z_{qq}+4Z_{pq} = 2Z_{pp}+Z_{qp}
    \end{equation}
    which implies the following relation for the respective anomalous dimensions:
    \begin{equation}\label{OME_gm}
        2\gamma_{qq}+4\gamma_{pq} = 2\gamma_{pp}+\gamma_{qp} 
    \end{equation}
\end{restatable}
\begin{restatable}{cor}{corone}\label{cor:peqq}
    If in the classical theory $p=2q$, then  $p_0=2q_0$, and at all scales $p(\mu)=2q(\mu)$ for any moment $N$.
\end{restatable}
\begin{restatable}{cor}{cortwo}\label{green_ward}
    For $p=2q$, the inverted bare 1PI momentum-space Green's functions for the operator matrix elements for any moment $N$ are related as:
    \begin{equation}\label{bare_ward}
        1-\Gamma^{\text{bare,1PI}}_{\phi|p=1}(Q^2) = 1-2\Gamma^{\text{bare,1PI}}_{\psi|p=1}(Q^2)
    \end{equation}
    where the subscript $|p=1$ means that the coupling $p$ has been normalised to $1$.
\end{restatable}

\subsection{Operator anomalous dimensions from $\lgen$}

Here we list the anomalous‑dimension expressions for the singlet operators from the generalised Lagrangian at general moment $N$ up to two loops, before demonstrating emergent SUSY in these operators. 
At this loop order the general expressions are straightforward to obtain, since the relevant functional form is already known from the results of \cite{Manashov:2025kgf} for the GNY model.

For the generalised Lagrangian $\lgop$, the one-loop OME anomalous dimensions as functions of the spin $N$ are given as:
\begin{align}
    \gamma_{qq}^{(1)}(N) = a\left(n_s-\frac{2n_s}{N(N+1)}\right)\, ,\quad & \quad \gamma_{pq}^{(1)}(N) = a\left(-\frac{n_s}{(N+1)}\right)\, , \nonumber\\[7pt]
    \gamma_{qp}^{(1)}(N) = a\left(\frac{-8n_f}{N}\right)\, ,\qquad\qquad\quad & \quad \gamma_{pp}^{(1)}(N) = 2n_fa\, ,
\end{align}
while the two‑loop results are expressed in terms of harmonic sums $S_{\vec{m}}(N)$, see~\cite{Vermaseren:1998uu},
with their argument omitted for brevity below:
\begin{align}
    \gamma_{qq}^{(2)}(N) &=  - \dfrac{2n_sa^2}{N(N + 1)}\Bigg((4-n_s)S_1 + \dfrac{n_s(1 + 2N)}{N^2(N + 1)^2} 
    -4(3-n_s)  - 5n_f
    \notag \\
    &\phantom{=  - \dfrac{2n_su^2}{N(N + 1)}\Bigg(}+ \dfrac{(4-3n_s) + (8-3n_s)N + n_f(3 + 2N)}{N(1 + N)}\Bigg)+\gamma_f^{(2)},\notag \\
    \gamma_{pq}^{(2)}(N) &= -
    \frac{2n_sa^2}{(N + 1)}\Bigg((2-n_s + n_f)S_1 + \dfrac{n_s(1 + 2N)}{2N^2(N + 1)^2}- 2n_f - \frac{8-3n_s}{2}\notag\\
    & \phantom{=-
    \frac{4n_sa^2}{N(N + 1)}\Bigg(}  + \dfrac{8 - n_s + 2(8-3n_s + 2n_f)N}{4N(N+1)} \Bigg),
    \notag \\
    \gamma_{qp}^{(2)}(N) &= -  \frac{8n_f a^2}{N} \Bigg((4-n_s)S_1 + \dfrac{n_s(1 + 2N)}{N^2(N + 1)^2} + \dfrac{(8-5n_s)}{2N(N + 1)} - (8-n_s)\Bigg),
    \notag \\[2pt]
    \gamma_{pp}^{(2)}(N) &= 
   -\frac{4n_f a^2} {N(N+ 1)}\Biggl(\dfrac{n_s(1 + 2N)}{N(N + 1)} - 2\Biggr) - \dfrac{3\lambda^2(n_s+2)}{N(N + 1)}+\gamma_s^{(2)}.
\end{align}
We can see that the results for the GNY model as calculated in \cite{Manashov:2025kgf} can be obtained by substituting $n_s=1$, provided the following change of conventions are carried out:
\begin{align}
        n_f=2n_f\Big |_{\text{\cite{Manashov:2025kgf}}} \, ,\quad \lambda =\frac{1}{3} \lambda\Big |_{\text{\cite{Manashov:2025kgf}}}\, ,\quad \gamma_{pq}^{(i)}(N) = \frac{N}{2}\gamma_{pq}^{(i)}(N)\Big |_{\text{\cite{Manashov:2025kgf}}} \, ,\quad \gamma_{qp}^{(i)}(N)=\frac{2}{N}\gamma_{qp}^{(i)}(N)\Big |_{\text{\cite{Manashov:2025kgf}}} 
\, .
\end{align}

\subsection{Emergent SUSY in OMEs}

In Sec.~\ref{sec:wardi} we obtained the relevant SUSY Ward identities that we expect to emerge from our generalised Lagrangian $\lgop$ at all even moments $N$. 
Note that, while eq.~\eqref{OME_gm} holds for any general value of $p$ and $q$, the other eq.~\eqref{bare_ward} only holds with $p=2q$. 
In this section, we verify these identities up to three loops for several even moments, thereby demonstrating that SUSY among the operators emerges from $\lgop$.

For the second Mellin moment $N=2$, the SUSY Ward identity \eqref{OME_gm} holds for all loop orders independent of the flavour content $(n_s,n_f)$. 
This is because for $N=2$ the sum of the operators reduces to the stress–energy tensor, which is a conserved Noether current and therefore does not undergo renormalisation, implying $p_0=p, q_0=q$.

If we substitute the two emergent‑SUSY points of $\lgen$ in the two‑loop results of the previous section, namely $n_s=1,n_f=\frac{1}{2}$ and $n_s=2,n_f=1$, 
the SUSY Ward identity \eqref{OME_gm} is satisfied.
The explicit check is provided in the anciliary file. 
Moreover, as in the case of the field anomalous dimensions, the highest transcendental weight at a given loop order requires only the equality of bosonic and fermionic degrees of freedom for emergent SUSY to appear:
\begin{itemize}
    \item one-loop case: $$(n_s-2n_f)\frac{N-2}{N} = 0$$
    \item two-loop term $\dfrac{S_1}{N(N+1)}$: $$(n_s-2n_f)(N-(4-n_s)(N+1))\dfrac{S_1}{N(N+1)}=0$$
    \item two-loop term $\dfrac{(2N+1)}{N^3(N+1)^3}$: $$(n_s-2n_f)\dfrac{(2N+1)}{N^3(N+1)^3}= 0$$
\end{itemize}

For the three‑loop case, the general expressions for the operator anomalous dimensions at arbitrary moment $N$ are not yet available. Indeed, in future work we will employ the emergent‑SUSY approach precisely to optimise the computation of these general results. For the present analysis, we have evaluated the OME anomalous dimensions at the even moments $N \leq 10$.

Recall from the discussion in Sec.~\ref{dimregsec} that, in odd dimensions, the trace of an odd number of $\gamma$-matrices does not generally vanish. 
Since the point $n_s=1,n_f=\frac{1}{2}$ corresponds to a three‑dimensional SUSY, using the four‑dimensional Clifford algebra does not reproduce the operator SUSY Ward identities at this point at three loops. 
At lower loop orders this issue did not arise, because no diagram produced a surviving Levi‑Civita tensor from the odd trace in eq.~\eqref{eq:oddtracethree}; as a result, even the incorrect Clifford algebra yielded consistent identities. 
At three loops, we must additionally account for the diagrams that contain a non‑vanishing trace of an odd number of $\gamma$‑matrices in order to recover emergent SUSY at the three‑dimensional WZ point:
\begin{equation}
\nonvqua    
\end{equation}
Recall that solid lines represent fermions (with arrows indicating the direction of flow) and dashed lines represent scalars. The letter markings indicate the specific lines on which the operator vertex is inserted. 
Placing the operator vertex on the two unmarked lower fermion lines yields a vanishing integral. Other diagrams with a non‑vanishing odd trace also exist, but their integrals likewise vanish. Only diagrams of the type shown above give a non‑zero contribution. Including these contributions restores the Ward identities \eqref{OME_gm} and \eqref{bare_ward} at three loops and for all even values of $N$ at both emergent‑SUSY points $n_s=1,n_f=\frac{1}{2}$ and $n_s=2,n_f=1$.

Beyond demonstrating emergent SUSY, the above calculations also verify the theorems from the previous section for the moments computed up to three loops. First, the SUSY relation among the anomalous dimensions is satisfied without imposing $p=2q$, in agreement with theorem~\ref{the_theorem}. 
Conversely, we have confirmed that the Ward identity for the Green’s functions, eq.~\eqref{bare_ward}, holds only after imposing $p=2q$, as required by corollary~\ref{green_ward}. This provides strong support for the theorems and their proofs.

\subsection{Operator renormalisation with emergent SUSY}

For the purposes of the present work, where our goal is to optimise the three‑loop singlet operators, it is sufficient to restrict the analysis to the three‑loop level. We therefore focus on formulating the application of emergent SUSY to the calculation of the operator anomalous dimensions, following the steps outlined in Sec.~\ref{sec:application}. In this context, the relevant Ward identity for optimisation is eq.~\eqref{bare_ward} for the bare 1PI correlators, since it allows us to perform the optimisation directly, without first renormalising.

The computation of flavour factors for the Green's functions in $\lgop$ is extremely fast using \texttt{FORM}, where we employ the $\gamma$‑matrix trace algorithms to evaluate the  $Z^I$ flavour factors. This calculation needs to be done only once, as the relevant graphs are the same for all moments. Consequently, the computational overhead introduced by using $\lgen$ is essentially negligible, making optimisation via emergent SUSY fully practical.

Take, for example, the two‑loop case, and apply step 1 from Sec.~\ref{sec:application}:
\begin{equation}
    \begin{aligned}
        \Gamma_\psi^{(2)} &= a_0^2\left(n_s(n_s-2)(p_0I_\psi^{1,p}+q_0I_\psi^{1,q})+n_fn_s(p_0I_\psi^{2,p}+q_0I_\psi^{2,q}\,)\right)\, ,\\
        \Gamma_\phi^{(2)} &= a_0^2\left(n_f(n_s-2)(p_0I_\phi^{1,p}+q_0I_\phi^{1,q})+n_fn_sq_0I_\phi^2\right) + \lambda_0^2(n_s+2)p_0I^3_\phi\;.
    \end{aligned}
\end{equation}
The coefficients of the flavour factors, the $I_\psi$'s and $I_\phi$'s, correspond to sums of integrals with the same flavour factor.  
Applying step 2, and using $a_0=\lambda_0$ together with $p_0 = 2q_0$, we obtain the following two equations for the two SUSY points from eq.~\eqref{bare_ward}:
\begin{equation}
    \begin{aligned}
         I_\phi^2 + 8I^3_\phi&= 2I_\psi^{2,p}+I_\psi^{2,q} \, ,\\
        \left((2I_\phi^{1,p}+I_\phi^{1,q})-I_\phi^2\right) - 24I^3_\phi  &=4(2I_\psi^{1,p}+I_\psi^{1,q})-(2I_\psi^{2,p}+I_\psi^{2,q}\,) \;.
    \end{aligned}
\end{equation}
Analogously to the two‑loop example in Sec.~\ref{sec:application}, we now apply steps 3 and 4 to eliminate two of the $I$'s, thereby reducing the number of integrals that must be evaluated.

At the three‑loop level the optimisation provided by emergent SUSY becomes especially significant. For higher Mellin moments, direct computation requires several days. 
To reconstruct the full $N$ dependence from a number of fixed moments, by solving Diophantine equations~\cite{Velizhanin:2010cm}, one must compute all moments up to at least $N \leq 30$. 
A full calculation of the moment $N=30$ takes already several days on a standard server, whereas emergent SUSY reduces this time by about 25\%.

This speed-up is based on the observation mentioned in Sec.~\ref{sec:application}, that the three-loop non-planar topology satisfies the SUSY Ward identity independently, without mixing with other topologies. And at higher moments it is the evaluation of the non-planar topology that takes almost the entire time of the full evaluation. 
The bare non-planar contributions to the  bare 1PI correlators are given as:
\begin{equation}
    \begin{aligned}
        \Gamma_\psi^{(3, \text{nonp})} &= a_0^2\,n_s(n_s^2-6n_s+4)(p_0I_\psi^{p}+q_0I_\psi^{q})\, ,\\
        \Gamma_\phi^{(3, \text{nonp})} &= a_0^2\,n_f(n_s^2-6n_s+4)(p_0I_\phi^{p}+q_0I_\phi^{q})\;.
    \end{aligned}
\end{equation}
Both the emergent SUSY points give us the identical equation between the integrals:
\begin{equation}
    4I_\psi^{p}+2I_\psi^{q} = 2I_\phi^{p}+I_\phi^{q}
\, ,
\end{equation}
which can be used to eliminate one of the four integral families. Thus, only 75\% of the heavy integrals need to be reduced via integration by parts, yielding an approximately 25\% speed‑up, as confirmed by explicit calculation.

This provides a proof‑of‑concept for using emergent SUSY, obtained through the method of generalised Lagrangians, to optimise computations in non‑supersymmetric theories. In the phenomenologically relevant case of QCD operator anomalous dimensions, where high‑moment calculations may take months, a 25\% reduction in computing time would correspond to saving several months of work. In this sense, SUSY, despite not being realised phenomenologically, serves as a powerful mathematical tool for improving phenomenological calculations.


\section{Conclusion and future outlooks}
\label{sec:concl}

In this work we established a framework for using emergent SUSY as a tool to optimise calculations in non‑supersymmetric theories. To develop this framework, we used the GNY model as a laboratory, since it is structurally simpler than QCD, which is our ultimate target. The central step is the construction of a generalised Lagrangian that unifies the non‑SUSY theory with as many SUSY theories as possible. 
This allows us to import their respective Ward identities, providing additional relations that reduce the number of integrals appearing in perturbative computations. 
The novelty of the generalised Lagrangian $\lgen$ lies not in its overall form, but in the flavour‑algebra relations given in eqs.~\eqref{eq:flavclifford} and \eqref{mijkl}. 
These define the minimal algebraic structure required for the Lagrangian to exhibit the desired interpolating behaviour between SUSY and non‑SUSY theories.

The specific application considered in this paper was the computation of fixed Mellin moments of the three‑loop anomalous dimensions of flavour‑singlet operators in the GNY model. Within this framework, we derived the SUSY Ward identity relevant for the renormalisation of these operators and found that certain non‑planar topologies satisfy this identity independently of the rest. Exploiting this property resulted in substantial gains in computational speed for the three‑loop calculation, with the improvement becoming even more pronounced at higher Mellin moments.

A central long‑term aim is to extend this strategy to QCD. The singlet‑sector Ward identity we analysed is the analogue of the relation obtained when QCD is embedded into an $\mathcal{N}=1$ SYM theory, and the emergent‑SUSY framework developed here now provides a systematic foundation for deriving and generalising such relations at higher loops.
The same methods also address a long‑standing issue in supersymmetric gauge theories. By applying the flavour‑algebra framework developed here, we have demonstrated that $\mathcal{N}=2$ SYM remains supersymmetric at three loops in dimensional reduction, in contrast to earlier claims that arose from subtleties in the treatment of flavour Levi‑Civita structures.

In contrast to the WZ model, SYM theories are nontrivially affected by dimensional reduction. In this scheme, the ghost fields couple only to the $4-2\eps$ degrees of freedom of the gluon, while the remaining $2\eps$ components behave as scalars. Owing to this mismatch, it had long been conjectured that dimensional reduction violates SUSY at sufficiently high loop order~\cite{Siegel:1979wq,Siegel:1980qs}. 
Existing calculations \cite{Avdeev:1982xy,Velizhanin:2008rw} appeared to support this view by indicating a SUSY violation at three loops in $\mathcal{N}=2$ SYM, despite the absence of such effects in $\mathcal{N}=1$ and $\mathcal{N}=4$ SYM. 
Using the techniques developed in this work, we showed in \cite{n2symat3loops} that SUSY is in fact preserved in $\mathcal{N}=2$ SYM at three loops, relying on consistent handling of flavour Levi‑Civita tensors.

These results open several promising avenues for further investigation. The observation that certain diagram topologies satisfy their own independent Ward identities suggests a deeper structural connection to the underlying master integrals. More broadly, the generalised‑Lagrangian approach raises the possibility that all‑loop SUSY results, such as the NSVZ $\beta$ function~\cite{Novikov:1983uc}, may become powerful tools for organising and constraining computations in non‑SUSY theories as well. Taken together, this work lays the foundation for a new conceptual and computational framework with wide‑ranging potential applications in perturbative quantum field theory.

\vspace{2mm}
\section*{Acknowledgments}
We would like to thank John Gracey for useful discussions.
The work has been supported by the European Research Council 
through  ERC Advanced Grant 101095857, {\it Conformal-EIC}.

\appendix
\section{Ward Identities}
\label{app:wi}

The modified Lagrangian for the WZ model with the operator vertices is given by 
\begin{equation}
    \mathcal{L}^{\text{WZ}}_{\text{op},N} = \mathcal{L}^{\text{WZ}}+\underbrace{ p_0\mathcal{O}^{\text{WZ}}_\phi(N)+q_0\mathcal{O}^{\text{WZ}}_\psi(N)}_{\mathcal{L}^{\text{WZ}}_{\mathcal{O}}}\, .
\end{equation} 
Then, the path integral with sources for this model is given by 
\begin{equation}
Z(J) = \int d\Phi\; e^{i(S+S_{\mathcal{O}}+\int d^Dx\, J^T\Phi)}\;,
\end{equation} 
where the actions, fields and source terms are:
\begin{itemize}
    \item $S=\int d^dx\,\mathcal{L}^{\text{WZ}}$ and $S_{\mathcal{O}}=\int d^Dx\,\mathcal{L}^{\text{WZ}}_{\mathcal{O}}$,
    \item $\Phi=(\phi,\phi^\dagger,\bar\psi_{\dot\alpha},\psi_\alpha)$ and $J =(J_\phi,J_{\phi^\dagger},-J_{\bar\psi}^{\dot{\alpha}},J_{\psi}^\alpha)$.
\end{itemize}
The Ward identity for any continuous transformation of the fields is obtained from the fact that the path integral $Z(J)$ does not change from a re-parametrization of the fields. 
Thus,
\begin{equation}\label{ward:app}
    \delta Z(J)=\langle\;\delta\mathcal{J} + i\delta S + i\delta S_{\mathcal{O}} + i\int d^Dx\;J^T(x)\delta\Phi(x)\;\rangle_J = 0
\, ,
\end{equation}
where $\mathcal{J}$ is the Jacobian of the field transformations. 

For the case of SUSY transformations, we know that $\delta J=0$. From the SUSY invariance of the WZ model, $\delta S = 0$. Next, as $p_0$ and $q_0$ are multiplied by the same Grassmann number, $S_{\mathcal{O}}^n=0 \;\forall n\geq 2$. Thus, we have:
\begin{align}
    \langle \;\int d^Dx\,J^T(x)\delta\Phi(x)\;\rangle_J &= \int d\Phi\;\int d^Dx\,J^T(x)\delta\Phi(x)\left(1+iS_\mathcal{O}\right)e^{i(S+\int d^Dx\, J^T\Phi)}\nonumber\\[4pt]
   \label{source} &=\langle\;\int d^Dx\,J^T(x)\delta\Phi(x)\;\rangle_{\small J}^{\small\text{WZ}} + \langle\;\int d^Dx\,J^T(x)\delta\Phi(x).iS_\mathcal{O}\;\rangle_{\small J}^{\small\text{WZ}}\, ,
\end{align}
where the superscript WZ indicates that the correlator is based on the pure WZ action without the operator vertices. Note that the Ward identity for the pure WZ case is given by:
\begin{equation}
\langle\;\delta\mathcal{J} + i\delta S + i\int d^Dx\;J^T(x)\delta\Phi(x)\;\rangle_{\small J}^{\small\text{WZ}} = \langle\;i\int d^Dx\;J^T(x)\delta\Phi(x)\;\rangle_{\small J}^{\small\text{WZ}} = 0\;.
\end{equation}
Hence, the first term in eq.~\eqref{source} does not contribute to the Ward identity \eqref{ward:app}, which now reduces to:
\begin{equation}
    \langle\;i\delta S_{\mathcal{O}} + i\int d^Dx\;J^T(x)\delta\Phi(x).iS_\mathcal{O}\;\rangle_{\small J}^{\small\text{WZ}} = 0
\end{equation}
Differentiate with respect to $J_{\phi^\dagger}(y_1)$ and $J_{\bar{\psi}}^{\dot{\alpha}}(y_2)$, and set $J=0$ to get
\begin{equation}
    \langle\;\phi(y_1)\bar\psi_{\dot\alpha}(y_2)\delta_{\epsilon,\bar\epsilon} S_{\mathcal{O}}\;\rangle^{\small\text{WZ}} + \langle\;\phi(y_1)\delta_{\epsilon,\bar\epsilon}\bar\psi_{\dot\alpha}(y_2) S_{\mathcal{O}}\;\rangle^{\small\text{WZ}}+\langle\;\delta_{\epsilon,\bar\epsilon}\bar\phi(y_1)\psi_{\dot\alpha}(y_2) S_{\mathcal{O}}\;\rangle^{\small\text{WZ}} = 0\;,
\end{equation}
where the infinitesimal SUSY transformations for the fields are given by:
\begin{align}
    \delta_{\epsilon,\bar\epsilon}\,\phi &= \epsilon^\alpha\psi_\alpha\; ,\\ 
    \delta_{\epsilon,\bar\epsilon}\,\bar\psi_{\dot\alpha} &= -i(\epsilon\sigma^\mu)_{\dot\alpha}\partial_\mu\phi^\dagger -\bar\epsilon_{\dot\alpha}F^\dagger
\, .
\end{align}
Above, we have suppressed overall factors of $\sqrt2$ in the right hand sides as it does not matter for the Ward identity, which we also do for the variation of $S_\mathcal{O}$:
\begin{align}\label{var_action}
{\lefteqn{
    \delta_{\epsilon,\bar\epsilon}\,S_{\mathcal{O}} =}} \nonumber\\ &=\epsilon^\alpha\left[p_0\phi^\dagger(-i\partial_+)^N\psi_\alpha-q_0\phi^\dagger\sigma^\mu_{\alpha\dot\alpha}\bar{\sigma}_+^{\dot\alpha\beta}(-i\partial_+)^{N-1}(-i\partial_\mu)\psi_\beta+q_0(\sigma_+)_{\alpha\dot\alpha}\bar\psi^{\dot\alpha}(-i\partial_+)^{N-1}F\right]\nonumber \\
    &-\bar\epsilon^{\dot\alpha}\left[p_0\bar\psi_{\dot\alpha}(-i\partial_+)^N\phi -q_0 \bar\psi_{\dot\beta}\bar{\sigma}_+^{\dot\beta\alpha}\sigma^\mu_{\alpha\dot\alpha}(-i\partial_+)^{N-1}(-i\partial_\mu)\phi+q_0F^\dagger(-i\partial_+)^{N-1}\psi^\alpha(\sigma_+)_{\alpha\dot\alpha}\right]
\, .
\end{align}
Notice that in the variation of the operator action, new operators with open fermionic indices emerge:
\begin{align}
    \mathcal{O}^1_{\alpha}& = \phi^\dagger(-i\partial_+)^N\psi_\alpha \, ,\\
    \mathcal{O}^2_{\alpha}& = \phi^\dagger(-i\partial_+)^{N-1}(-i\partial_\mu)(\sigma^\mu\bar{\sigma}_+\psi)_\alpha \, ,\\
    \mathcal{O}^F_{\alpha}& =(\sigma_+\bar\psi)_\alpha(-i\partial_+)^{N-1}F \, .
\end{align}
The operators appearing in the coefficient of $\bar\epsilon^{\dot\alpha}$ are just the Hermitian conjugates of the above operators. Now consider the Ward identity for the coefficient of $\epsilon^\alpha$:
\begin{equation}\label{ward_explicit:app}
\begin{aligned}
    0  =&\:(-i\slashed\partial_{y_2})_{\alpha\dot\alpha}\langle\phi(y_1)\,S_\mathcal{O}\,\phi^\dagger(y_2)\rangle + \langle\psi_\alpha(y_1)\,S_\mathcal{O}\,\bar\psi_{\dot\alpha}(y_2)\rangle \\
    &+\langle\phi(y_1)\int d^Dx\left[p_0\mathcal{O}^1_{\alpha}(x)+q_0\mathcal{O}^2_{\alpha}(x)-q_0\mathcal{O}^F_{\alpha}(x)\right]\,\bar\psi_{\dot\alpha}(y_2)\rangle\;.
\end{aligned}
\end{equation}
The two correlators in the first line of eq.~\eqref{ward_explicit} are the ones for which we seek to establish a relation. However, the correlator in the second line, involving the exotic operators, does not vanish and in fact contains even more complicated operator structures. Thus, this equation does not yield a useful emergent‑SUSY relation between the OMEs.

\section{Light-cone SUSY}
\label{app:lightconesusy}

In the light-cone basis, the light-cone $\sigma$ matrices are given by 
\begin{equation}
    \sigma_\pm = \frac{1}{\sqrt2}(\sigma^0 \pm i\sigma^3)\, , \qquad \bar\sigma_\pm = \frac{1}{\sqrt2}(\bar\sigma^0 \pm i\bar\sigma^3)\, ,
\end{equation}
while the other two matrices are left unchanged. The Clifford algebra is given by the usual $\sigma^\mu\bar\sigma^\nu+\sigma^\nu\bar\sigma^\mu = 2g^{\mu\nu}$ where now the Lorentz indices are $(+,-,1,2)$, and the metric tensor given by 
\begin{equation}
 g^{\mu\nu}=
\begin{pmatrix}
0 & 1 & 0 & 0\\
1 & 0 & 0 & 0\\
0 & 0 & -1 & 0\\
0 & 0 & 0 & -1
\end{pmatrix}   
\, .
\end{equation}
Two pairs of light-cone projection matrices can thus be defined:
\begin{equation}
    P_{\pm} =\frac{1}{2}\sigma_\pm\bar\sigma_\mp\, ,\qquad \bar P_{\pm} =\frac{1}{2}\bar\sigma_\mp\sigma_\pm\, ,
\end{equation} 
based on which any spinor can be written as $\chi=\chi_++\chi_-$ (same for $\bar\chi$) where $\chi_\pm=P_{\pm}$ and $\bar\chi_\pm=\bar P_{\pm}\bar\chi$. Based on these definitions, we have the following properties:
\begin{align}
    \bar{\sigma}_+\chi_+ & = \frac{1}{2}\bar\sigma_+\sigma_+\bar\sigma_-\chi = 0\, , \\
    \bar\chi_+\bar{\sigma}_+ & =  \frac{1}{2}\sigma_-\bar\sigma_+\bar\chi\,\bar\sigma_+ = \frac{1}{2}\sigma_-\bar\sigma_+\sigma_+\bar\chi = 0 \, ,
\end{align}
where $\bar\sigma_+\sigma_+ = 0$ from the Clifford algebra relation, and for the second property we used spinor identities. Likewise, 
\begin{equation}
\chi_+{\sigma}_+ = 0\, ,
\qquad{\sigma}_+\bar\chi_+=0 \, .
\end{equation}
This means that, 
\begin{equation}
\mathcal{O}^{\text{WZ}}_\psi(N) = \bar\psi\,\overline\sigma_+(-i\partial_+)^{N-1}\psi = \bar\psi_-\,\overline\sigma_+(-i\partial_+)^{N-1}\psi_-\,,   
\end{equation}
signifying that the $\psi_+$ component does not contribute to the operator.

As the $+$ and $-$ spinors are independent and form a basis, we can choose $\epsilon_+,\epsilon_-$ (same for $\bar\epsilon$) as the independent SUSY transformation parameters. Then, considering the $\epsilon_+$ component of eq.~\eqref{var_action}, we have:
\begin{align}
    &p_0\phi^\dagger(-i\partial_+)^N\epsilon_+\psi_--q_0\phi^\dagger\epsilon_+\sigma^\mu\bar{\sigma}_+(-i\partial_+)^{N-1}(-i\partial_\mu)\psi_-+q_0\epsilon_+\sigma_+\bar\psi_-(-i\partial_+)^{N-1}F \nonumber \\
    &=p_0\phi^\dagger(-i\partial_+)^N\epsilon_+\psi_-- q_0\phi^\dagger\epsilon_+n_\nu(\sigma^\mu\bar{\sigma}^\nu)(-i\partial_+)^{N-1}(-i\partial_\mu)\psi_- \nonumber \\
    &=p_0\phi^\dagger(-i\partial_+)^N\epsilon_+\psi_-- q_0\phi^\dagger\epsilon_+n_\nu(2g^{\mu\nu}-\sigma^\nu\bar{\sigma}^\mu)(-i\partial_+)^{N-1}(-i\partial_\mu)\psi_- \nonumber \\
    &=(p_0-2q_0)\phi^\dagger(-i\partial_+)^N\epsilon_+\psi_-+ q_0\phi^\dagger\epsilon_+\sigma_+\bar{\sigma}^\mu(-i\partial_+)^{N-1}(-i\partial_\mu)\psi_- \nonumber \\
    &=\epsilon_+(p_0-2q_0)\mathcal{O}^1_-
\, ,
\end{align}
where we used $\epsilon_+\sigma_+=0$ to eliminate the terms. 
Then, the Ward identity \eqref{ward_explicit:app} becomes
\begin{equation}\label{ward_ult:app}
\begin{aligned}
     0=&-i\sigma_-\partial_+\langle\phi(y_1)\,S_\mathcal{O}\,\phi^\dagger(y_2)\rangle + \langle\psi_-(y_1)\,S_\mathcal{O}\,\bar\psi_-(y_2)\rangle\\
    &+\langle\phi(y_1)\int d^Dx(p_0-2q_0)\mathcal{O}^1_-(x)\,\bar\psi_-(y_2)\rangle\;.
\end{aligned}
\end{equation}
Thus we have gotten rid of both the operators $\mathcal{O}^2_\alpha$ and $\mathcal{O}^F_\alpha$. This form of the Ward identity gives us useful relations among the Green's functions and anomalous dimensions of the OMEs. Recall that 
\begin{align}
    p_0 &= \mu^{2-N}(pZ_{pp} + q Z_{qp}) \, ,\\
    q_0 &= \mu^{2-N}(pZ_{pq} + q Z_{qq}) \, ,
\end{align}
where the $Z_{ij}$'s only depend on the square of the coupling constant $\alpha = g^2$, but not on $p$ and $q$. 
Using the Ward identity as given in eq.~\eqref{ward_ult:app}, we propose the following statement:
\mainthm*
\begin{proof}
    As the correlators in Ward identities are always the renormalised versions, when we can act both sides of eq.~\eqref{ward_ult:app} by the Rota-Baxter projector $R$ corresponding to the chosen renormalisation scheme, the correlators corresponding to the OMEs vanish, and we have:
\begin{equation}\label{rota}
    R\left[(p_0-2q_0)\langle\phi(y_1)\int d^Dx\,\mathcal{O}^1_-(x)\,\bar\psi_-(y_2)\rangle\right] = 0\,.
\end{equation}    
    Thus, $(p_0-2q_0)$ must give the correct counterterms to renormalise the singularities created by $\mathcal{O}^1_-(x)$. Before we proceed, note the following properties of the $Z_{ij}$'s:
\begin{equation}
 Z_{pp}, Z_{qq} = 1+\mathcal{O}(\alpha)\, ,\qquad 
 Z_{pq}, Z_{qp} = \mathcal{O}(\alpha)\,.
\end{equation}    
Thus,
\begin{equation*}
        p_0-2q_0 = \mu^{2-N}(p-2q +p\alpha(Z^{(1)}_{pp}-2Z^{(1)}_{pq}) +q\alpha(Z^{(1)}_{qp}-2Z^{(1)}_{qq}) + \mathcal{O}(\alpha^2))\;,
    \end{equation*}
where $Z^{(1)}$ is the one-loop contribution to the counterterm. We suppressed the $\mu$ dependence of the renormalised couplings for brevity. 
On the other hand, 
\begin{equation}
     \langle\phi(y_1)\int d^Dx\,\mathcal{O}^1_-(x)\,\bar\psi_-(y_2)\rangle = Z_\phi^{-1}Z_\psi^{-1}\frac{p_+\slashed{p}}{(p^2)^2}\Big(1+\alpha I^{(1)} + \mathcal{O}(\alpha^2)\Big)\, ,
    \end{equation}
    where $Z_\phi$ and $Z_\psi$ are the field renormalisations, and $I^{(1)}$ is the one-loop contribution.  
    Using this, the coefficient of the first order in $\alpha$ of eq.~\eqref{rota} is given by:
    \begin{align}\label{ind_start}
        (p-2q)R[I^{(1)}] +pR[Z^{(1)}_{pp}-2Z^{(1)}_{pq}] + qR[Z^{(1)}_{qp}-2Z^{(1)}_{qq}] &= 0 \nonumber \\
        \implies \quad (p-2q)R[I^{(1)}] +p(Z^{(1)}_{pp}-2Z^{(1)}_{pq}) + q(Z^{(1)}_{qp}-2Z^{(1)}_{qq}) &=0
        \, ,
    \end{align}
    where the second line holds because $R[Z]=Z$ for any counterterm in the $\overline{\text{MS}}$-scheme.

    Note that eq.~\eqref{ind_start} holds for any value for $p(\mu)$ and $q(\mu)$. Thus the equation will hold even if we choose at the classical level that $p=2q$. Then there exists a renormalisation scale $\mu_1$ where  
    \begin{equation}\label{substi}
        p(\mu_1)=2q(\mu_1) \;.
    \end{equation}
    As the $Z_{ij}$ factors do not depend on $p$ and $q$, their values are independent of the imposition of the above relation, and thus eq.~\eqref{ind_start} becomes:
    \begin{eqnarray}\label{peetoq}
        2Z^{(1)}_{pp}-4Z^{(1)}_{pq} = Z^{(1)}_{qp}-2Z^{(1)}_{qq}\phantom{1}\\
        \implies \quad 
        2Z^{(1)}_{qq}+4Z^{(1)}_{pq} = 2Z^{(1)}_{pp}+Z^{(1)}_{qp}\, .
    \end{eqnarray}
    This is the one-loop contribution to eq.~\eqref{OME_ctm} that we seek to prove. Hence, we can build up an inductive proof. Let the inductive hypothesis be that the equation 
\begin{equation}
     2Z^{(j)}_{qq}+4Z^{(j)}_{pq} = 2Z^{(j)}_{pp}+Z^{(j)}_{qp}   
\end{equation}
    holds for all $j\leq n$. Now for $j=n+1$, consider the coefficient of $\alpha^{n+1}$ in eq.~\eqref{rota}:
    \begin{align}
        0 =&\sum_{j=0}^{n+1}I^{n+1-j}\left(p(Z^{(j)}_{pp}-2Z^{(j)}_{pq}) + q(Z^{(j)}_{qp}-2Z^{(j)}_{qq})\right) \nonumber \\
        =& \sum_{j=0}^{n}I^{n+1-j}(p-2q)(Z^{(j)}_{pp}-2Z^{(j)}_{pq}) \nonumber \\
        &+ p(Z^{(n+1)}_{pp}-2Z^{(n+1)}_{pq}) + q(Z^{(n+1)}_{qp}-2Z^{(n+1)}_{qq}) \nonumber \\[5pt]
        =&-2q(Z^{(n+1)}_{pp}-2Z^{(n+1)}_{pq}) + q(Z^{(n+1)}_{qp}-2Z^{(n+1)}_{qq})\;,
    \end{align}
    where in the last line we substituted eq.~\eqref{substi}. 
    Thus we have for $j=n+1$: 

\begin{equation}
    2Z^{(n+1)}_{qq}+4Z^{(n+1)}_{pq} = 2Z^{(n+1)}_{pp}+Z^{(n+1)}_{qp}\;.
\end{equation}
    Thus by induction, this holds for all $j$ and the full counterterms satisfy eq.~\eqref{OME_ctm}, from which eq.~\eqref{OME_gm} follows.
\end{proof}
\corone*
\begin{proof}
    If $p=2q$ in the classical theory, for some scale $\mu_1$ we have $p(\mu_1)=2q(\mu_1)$. Then, consider 
    \begin{align}
        p_0-2q_0 & = \mu_1^{2-N}\left(p(\mu_1)(Z_{pp}-2Z_{pq}) + q(\mu_1)(Z_{qp}-2Z_{qq})\right) \nonumber \\
        &=\mu_1^{2-N}(p(\mu_1)-2q(\mu_1))(Z_{pp}-2Z_{pq}) \nonumber \\[3pt]
        \implies \quad p_0&= 2q_0\;. 
    \end{align}
    Here we used theorem \ref{the_theorem} in the second line. As bare couplings do not depend on the renormalisation scale, this relation holds for any $\mu$. Now we invert relation \eqref{oper_ren} to get 
    \begin{equation}
    \begin{bmatrix}
        p \\
        q
    \end{bmatrix} =\frac{\mu^{N-2}}{\det Z}
    \begin{bmatrix}
        Z_{qq} & -Z_{qp} \\
        -Z_{pq} & Z_{pp}
    \end{bmatrix}
    \begin{bmatrix}
        p_0 \\
        q_0
    \end{bmatrix}
    \end{equation}
    where $\det Z$ is the determinant of the matrix. This implies that, for any renormalisation scale $\mu$, 
    \begin{align}
        p(\mu)-2q(\mu) &= \frac{\mu^{N-2}}{\det Z}\left(p_0(Z_{qq}+2Z_{pq})-q_0(2Z_{pp}+Z_{qp})\right) \nonumber \\
        &= \frac{\mu^{N-2}}{\det Z}\left(p_0-2q_0\right)(Z_{qq}+2Z_{pq}) = 0
    \end{align}
    which is what we sought to prove.
\end{proof}
\cortwo*
\begin{proof}
    As due to the previous corollary, $p_0=2q_0$, the Ward identity \eqref{ward_ult:app} reduces to 
    \begin{align}
    0=&-i\sigma_-\partial_+\langle\phi(y_1)\,S_\mathcal{O}\,\phi^\dagger(y_2)\rangle + \langle\psi_-(y_1)\,S_\mathcal{O}\,\bar\psi_-(y_2)\rangle
\, .
    \end{align} 
    In the momentum space, we then have
    \begin{align}
        Z_\phi^{-2}\,p_0\frac{\sigma_-Q_+^{N+1}}{(Q^2)^2\left(1-\frac{1}{p_0}\Gamma^{\text{bare,1PI}}_\phi(Q^2)\right)} - Z_\psi^{-2}\,q_0\frac{\sigma_-\bar\sigma_+\sigma_-Q_+^{N+1}}{(Q^2)^2\left(1-\frac{1}{q_0}\Gamma^{\text{bare,1PI}}_\psi(Q^2)\right)}  &= 0 \nonumber \\[3pt]
        \implies \quad p_0\frac{\sigma_-Q_+^{N+1}}{(Q^2)^2\left(1-\frac{1}{p_0}\Gamma^{\text{bare,1PI}}_\phi(Q^2)\right)} - 2q_0\frac{\sigma_-Q_+^{N+1}}{(Q^2)^2\left(1-\frac{1}{q_0}\Gamma^{\text{bare,1PI}}_\psi(Q^2)\right)}  &= 0
        \, ,
    \end{align}
    where $Z_\phi$ and $Z_\psi$ are the field renormalisations and are equal because of SUSY. 
    We then used the anti-commutation relations for light cone sigma matrices in the second line. 
Because of $p_0=2q_0$ we get 
\begin{equation}
        1-\frac{1}{p_0}\Gamma^{\text{bare,1PI}}_\phi(Q^2)= 1-\frac{1}{q_0}\Gamma^{\text{bare,1PI}}_\psi(Q^2)\,.
\end{equation}
As both $\Gamma$ terms are proportional to $p_0$ (for $q_0$'s we substitute $p_0/2$), we can factor out $p_0$ and thus obtain eq.~\eqref{bare_ward}.

\end{proof}

\end{document}

%% file: pics.tex
\newcommand{\olsclro}{\begin{minipage}{2.5cm}
\begin{tikzpicture}
\draw (-1, 0.2) node{${\scriptstyle I}$};
\draw (1, 0.2) node{${\scriptstyle J}$};
        \begin{feynman}
             \vertex (b) at (0.5,0);
             \vertex (d) at (-0.5,0);
             \vertex (e) at (-1, 0);
             \vertex (f) at (1, 0);
             \diagram* {
             (b) -- [fermion, half left] (d);
             (b) -- [fermion, half right] (d);
             (e) -- [scalar] (d);
             (f) -- [scalar] (b); 
             };
          \end{feynman}
\end{tikzpicture}
\end{minipage}}

\newcommand{\olsclrt}{\begin{minipage}{2.5cm}
\begin{tikzpicture}
\draw (-1, 0.2) node{${\scriptstyle I}$};
\draw (1, 0.2) node{${\scriptstyle J}$};
        \begin{feynman}
             \vertex (b) at (0.5,0);
             \vertex (d) at (-0.5,0);
             \vertex (e) at (-1, 0);
             \vertex (f) at (1, 0);
             \diagram* {
             (d) -- [fermion, half left] (b);
             (d) -- [fermion, half right] (b);
             (e) -- [scalar] (d);
             (f) -- [scalar] (b); 
             };
          \end{feynman}
\end{tikzpicture}
\end{minipage}}

\newcommand{\olfer}{\begin{minipage}{2.5cm}
\begin{tikzpicture}
\draw (-1, -0.2) node{${\scriptstyle A}$};
\draw (1, 0.2) node{${\scriptstyle B}$};
        \begin{feynman}
             \vertex (b) at (0.5,0);
             \vertex (d) at (-0.5,0);
             \vertex (e) at (-1, 0);
             \vertex (f) at (1, 0);
             \diagram* {
             (b) -- [scalar, half left] (d);
             (d) -- [fermion, half left] (b);
             (d) -- [fermion] (e);
             (f) -- [fermion] (b); 
             };
          \end{feynman}
\end{tikzpicture}
\end{minipage}}

\newcommand{\olyuk}{\begin{minipage}{2cm}
\begin{tikzpicture}
\draw (0, 1.45) node{${\scriptstyle I}$};
\draw (-0.9, -0.2) node{${\scriptstyle A}$};
\draw (0.9, -0.2) node{${\scriptstyle B}$};
        \begin{feynman}
             \vertex (b) at (0.5,0);
             \vertex (c) at (-0.5,0);
             \vertex (a) at (0,0.85);
             \vertex (d) at (0, 1.25);
             \vertex (e) at (-0.75, -0.42);
             \vertex (f) at (0.75, -0.42);
             \diagram* {
             (b) -- [scalar] (c);
             (a) -- [fermion] (b);
             (a) -- [fermion] (c);
             (e) -- [fermion] (c);
             (f) -- [fermion] (b);
             (a) -- [scalar] (d);
             };
          \end{feynman}
\end{tikzpicture}
\end{minipage}}

\newcommand{\tlyuk}{\begin{minipage}{2cm}
\begin{tikzpicture}
\draw (0, 1.45) node{${\scriptstyle I}$};
\draw (-1.65, -0.62) node{${\scriptstyle A}$};
\draw (1.65, -0.62) node{${\scriptstyle B}$};
        \begin{feynman}
             \vertex (b) at (0.5,0.21);
             \vertex (c) at (-0.5,0.21);
             \vertex (g) at (1,-0.42);
             \vertex (h) at (-1,-0.42);
             \vertex (a) at (0,0.84);
             \vertex (d) at (0, 1.25);
             \vertex (e) at (-1.5, -0.84);
             \vertex (f) at (1.5, -0.84);
             \diagram* {
             (b) -- [scalar] (h);
             (c) -- [scalar] (g);
             (b) -- [fermion] (g);
             (c) -- [fermion] (h);
             (b) -- [fermion] (a);
             (c) -- [fermion] (a);
             (e) -- [fermion] (h);
             (f) -- [fermion] (g);
             (a) -- [scalar] (d);
             };
          \end{feynman}
\end{tikzpicture}
\end{minipage}}

\newcommand{\olfouro}{\begin{minipage}{2cm}
\begin{tikzpicture}
\draw (-0.8, 0.5) node{${\scriptstyle I}$};
\draw (0.8, 0.5) node{${\scriptstyle J}$};
\draw (0.8, -1.1) node{${\scriptstyle K}$};
\draw (-0.8, -1.1) node{${\scriptstyle L}$};
        \begin{feynman}
             \vertex (a) at (0.5,0.17);
             \vertex (d) at (-0.5,0.17);
             \vertex (b) at (0.5,-0.83);
             \vertex (c) at (-0.5,-0.83);
             \vertex (g) at (0.75, 0.42);
             \vertex (h) at (-0.75, 0.42);
             \vertex (e) at (-0.75, -1.08);
             \vertex (f) at (0.75, -1.08);
             \diagram* {
             (g) -- [scalar] (a);
             (h) -- [scalar] (d);
             (e) -- [scalar] (c);
             (f) -- [scalar] (b);
             (a) -- [fermion] (d);
             (a) -- [fermion] (b);
             (c) -- [fermion] (b);
             (c) -- [fermion] (d);
             };
          \end{feynman}
\end{tikzpicture}
\end{minipage}}

\newcommand{\olfourt}{\begin{minipage}{2cm}
\begin{tikzpicture}
\draw (-0.8, 0.5) node{${\scriptstyle I}$};
\draw (0.8, 0.5) node{${\scriptstyle J}$};
\draw (0.8, -1.1) node{${\scriptstyle K}$};
\draw (-0.8, -1.1) node{${\scriptstyle L}$};
        \begin{feynman}
             \vertex (a) at (0.5,0.17);
             \vertex (d) at (-0.5,0.17);
             \vertex (b) at (0.5,-0.83);
             \vertex (c) at (-0.5,-0.83);
             \vertex (g) at (0.75, 0.42);
             \vertex (h) at (-0.75, 0.42);
             \vertex (e) at (-0.75, -1.08);
             \vertex (f) at (0.75, -1.08);
             \diagram* {
             (g) -- [scalar] (a);
             (h) -- [scalar] (d);
             (e) -- [scalar] (c);
             (f) -- [scalar] (b);
             (d) -- [fermion] (a);
             (b) -- [fermion] (a);
             (b) -- [fermion] (c);
             (d) -- [fermion] (c);
             };
          \end{feynman}
\end{tikzpicture}
\end{minipage}}

\newcommand{\olfourscl}{\begin{minipage}{2cm}
\begin{tikzpicture}
\draw (-0.8, 0.5) node{${\scriptstyle I}$};
\draw (0.8, 0.5) node{${\scriptstyle J}$};
\draw (0.8, -0.5) node{${\scriptstyle K}$};
\draw (-0.8, -0.5) node{${\scriptstyle L}$};
        \begin{feynman}
             \vertex (b) at (0.5,0);
             \vertex (c) at (-0.5,0);
             \vertex (g) at (0.75, 0.42);
             \vertex (h) at (-0.75, 0.42);
             \vertex (e) at (-0.75, -0.42);
             \vertex (f) at (0.75, -0.42);
             \diagram* {
             (g) -- [scalar] (b);
             (h) -- [scalar] (c);
             (e) -- [scalar] (c);
             (f) -- [scalar] (b);
             (b) -- [scalar, half left] (c);
             (b) -- [scalar, half right] (c);
             };
          \end{feynman}
\end{tikzpicture}
\end{minipage}}

\newcommand{\glsclr}{\begin{minipage}{2.5cm}
\begin{tikzpicture}
\draw (-1, 0.2) node{${\scriptstyle I}$};
\draw (1, 0.2) node{${\scriptstyle J}$};
        \begin{feynman}
             \vertex (b) at (0.5,0);
             \vertex (d) at (-0.5,0);
             \vertex (e) at (-1, 0);
             \vertex (f) at (1, 0);

             \diagram* {
             (e) -- [scalar] (d);
             (f) -- [scalar] (b); 
             };
             \end{feynman}
             \filldraw[fill=gray!50, draw=black] (0,0) circle (0.5cm);
\end{tikzpicture}
\end{minipage}}

\newcommand{\glfer}{\begin{minipage}{2.5cm}
\begin{tikzpicture}
\draw (-1, -0.2) node{${\scriptstyle A}$};
\draw (1, 0.2) node{${\scriptstyle B}$};
        \begin{feynman}
             \vertex (b) at (0.5,0);
             \vertex (d) at (-0.5,0);
             \vertex (e) at (-1, 0);
             \vertex (f) at (1, 0);
             \diagram* {
             (d) -- [fermion] (e);
             (f) -- [fermion] (b); 
             };
          \end{feynman}
          \filldraw[fill=gray!50, draw=black] (0,0) circle (0.5cm);
\end{tikzpicture}
\end{minipage}}

\newcommand{\glyuk}{\begin{minipage}{2cm}
\begin{tikzpicture}
\draw (0, 1.45) node{${\scriptstyle I}$};
\draw (-0.9, -0.2) node{${\scriptstyle A}$};
\draw (0.9, -0.2) node{${\scriptstyle B}$};
        \begin{feynman}
             \vertex (b) at (0.5,0);
             \vertex (c) at (-0.5,0);
             \vertex (a) at (0,0.85);
             \vertex (d) at (0, 1.25);
             \vertex (e) at (-0.75, -0.42);
             \vertex (f) at (0.75, -0.42);
             \diagram* {
             (e) -- [fermion] (c);
             (f) -- [fermion] (b);
             (a) -- [scalar] (d);
             };
          \end{feynman}
          \filldraw[fill=gray!50, draw=black] (0,0.28) circle (0.56cm);
\end{tikzpicture}
\end{minipage}}

\newcommand{\glyukconj}{\begin{minipage}{2cm}
\begin{tikzpicture}
\draw (0, 1.45) node{${\scriptstyle I}$};
\draw (-0.6, -0.4) node{${\scriptstyle A}$};
\draw (0.6, -0.4) node{${\scriptstyle B}$};
        \begin{feynman}
             \vertex (b) at (0.5,0);
             \vertex (c) at (-0.5,0);
             \vertex (a) at (0,0.85);
             \vertex (d) at (0, 1.25);
             \vertex (e) at (-0.75, -0.42);
             \vertex (f) at (0.75, -0.42);
             \diagram* {
             (c) -- [fermion] (e);
             (b) -- [fermion] (f);
             (a) -- [scalar] (d);
             };
          \end{feynman}
          \filldraw[fill=gray!50, draw=black] (0,0.28) circle (0.56cm);
\end{tikzpicture}
\end{minipage}}

\newcommand{\glfour}{\begin{minipage}{2cm}
\begin{tikzpicture}
\draw (-0.8, 0.5) node{${\scriptstyle I}$};
\draw (0.8, 0.5) node{${\scriptstyle J}$};
\draw (0.8, -1.1) node{${\scriptstyle K}$};
\draw (-0.8, -1.1) node{${\scriptstyle L}$};
        \begin{feynman}
             \vertex (a) at (0.35,0.02);
             \vertex (d) at (-0.35,0.02);
             \vertex (b) at (0.35,-0.68);
             \vertex (c) at (-0.35,-0.68);
             \vertex (g) at (0.75, 0.42);
             \vertex (h) at (-0.75, 0.42);
             \vertex (e) at (-0.75, -1.08);
             \vertex (f) at (0.75, -1.08);
             \diagram* {
             (g) -- [scalar] (a);
             (h) -- [scalar] (d);
             (e) -- [scalar] (c);
             (f) -- [scalar] (b);
             };
          \end{feynman}
          \filldraw[fill=gray!50, draw=black] (0,-0.33) circle (0.5cm);
\end{tikzpicture}
\end{minipage}}

\newcommand{\intpsio}{
\begin{minipage}{1.8cm}
\begin{tikzpicture}[scale=1.5]
    \begin{feynman}
             \vertex (vg1) at (-0.4,0);
             \vertex (vg2) at (0.4,0);
             \vertex (vq2) at (0,0.4);
             \vertex (vs) at (0,-0.4);
             \vertex (d) at (0.6,0);
             \vertex (e) at (-0.6,0);
             \diagram* {
             (e) --   (vg1);
             (d) --   (vg2);
             (vq2) --   (vg1);
             (vs) --   (vg2);
             (vs) -- (vq2);
             (vg1) -- [dashed, dash pattern=on 1.5pt off 1.5pt] (vs);
            (vq2) -- [dashed, dash pattern=on 1.5pt off 1.5pt] (vg2);
             };
    \end{feynman}
\end{tikzpicture}
    \end{minipage}
}

\newcommand{\intphio}{
\begin{minipage}{1.8cm}
\begin{tikzpicture}[scale=1.5]
    \begin{feynman}
             \vertex (vg1) at (-0.4,0);
             \vertex (vg2) at (0.4,0);
             \vertex (vq2) at (0,0.4);
             \vertex (vs) at (0,-0.4);
             \vertex (d) at (0.6,0);
             \vertex (e) at (-0.6,0);
             \diagram* {
             (vs) --   (vg1);
             (vq2) --   (vg2);
             (vq2) --   (vg1);
             (vs) --   (vg2);
             (vs) -- [dashed, dash pattern=on 1.5pt off 1.5pt] (vq2);
             (vg1) -- [dashed, dash pattern=on 1.5pt off 1.5pt] (e);
             (d) -- [dashed, dash pattern=on 1.5pt off 1.5pt] (vg2);
             };
    \end{feynman}
\end{tikzpicture}
    \end{minipage}
}

\newcommand{\intpsitw}{
\begin{minipage}{1.8cm}
\begin{tikzpicture}[scale=1.5]
\draw (0,0.5) node{};
\draw (0.4,0) arc(0:180:0.4);
\draw (0.4,0) [dashed, dash pattern=on 1.5pt off 1.5pt] arc(360:300:0.4);
\draw (-0.4,0) [dashed, dash pattern=on 1.5pt off 1.5pt] arc(180:240:0.4);
    \begin{feynman}
             \vertex (vg1) at (-0.4,0);
             \vertex (vg2) at (0.4,0);
             \vertex (vq2) at (-0.2,-0.34);
             \vertex (vq1) at (0.2,-0.34);
             \vertex (d) at (0.6,0);
             \vertex (e) at (-0.6,0);
             \diagram* {
             (e) --   (vg1);
             (d) --   (vg2);
             (vq2) -- [half left] (vq1);
             (vq2) -- [half right] (vq1);
             };
    \end{feynman}
\end{tikzpicture}
    \end{minipage}
}

\newcommand{\intphitw}{
\begin{minipage}{1.8cm}
\begin{tikzpicture}[scale=1.5]
\draw (0,0.5) node{};
\draw (0.4,0) arc(0:180:0.4);
\draw (0.4,0) arc(360:300:0.4);
\draw (-0.4,0)  arc(180:240:0.4);
\draw (0.4,0)  arc(360:300:0.4);
    \begin{feynman}
             \vertex (vg1) at (-0.4,0);
             \vertex (vg2) at (0.4,0);
             \vertex (vq2) at (-0.2,-0.34);
             \vertex (vq1) at (0.2,-0.34);
             \vertex (d) at (0.6,0);
             \vertex (e) at (-0.6,0);
             \diagram* {
             (e) --  [dashed, dash pattern=on 1.5pt off 1.5pt]  (vg1);
             (d) --   [dashed, dash pattern=on 1.5pt off 1.5pt] (vg2);
             (vq2) -- [half left] (vq1);
             (vq2) -- [dashed, dash pattern=on 1.5pt off 1.5pt, half right] (vq1);
             };
    \end{feynman}
\end{tikzpicture}
    \end{minipage}
}

\newcommand{\intphithr}{
\begin{minipage}{1.8cm}
\begin{tikzpicture}[scale=1.5]
\draw (0.4,0) [dashed, dash pattern=on 1.5pt off 1.5pt] arc(0:360:0.4);
    \begin{feynman}
             \vertex (vg1) at (-0.4,0);
             \vertex (vg2) at (0.4,0);
             \vertex (vq2) at (-0.2,-0.34);
             \vertex (vq1) at (0.2,-0.34);
             \vertex (d) at (0.6,0);
             \vertex (e) at (-0.6,0);
             \diagram* {
             (vg1) --   [dashed, dash pattern=on 1.5pt off 1.5pt] (vg2);
             (e) --  [dashed, dash pattern=on 1.5pt off 1.5pt]  (vg1);
             (d) --   [dashed, dash pattern=on 1.5pt off 1.5pt] (vg2);
             };
    \end{feynman}
\end{tikzpicture}
    \end{minipage}
}

\newcommand{\nonpphi}{
\begin{minipage}{3.5cm}
\begin{tikzpicture}[scale=2]
    \begin{feynman}
             \vertex (vg1) at (-0.6,0);
             \vertex (vg2) at (0.6,0);
             \vertex (vq1) at (-0.3,0.3);
             \vertex (vq2) at (-0.3,-0.3);
             \vertex (vq3) at (0.3,-0.3);
             \vertex (vq4) at (0.3,0.3);
             \vertex (v1) at (0.05,-0.05);
             \vertex (v2) at (-0.05,0.05);
             \vertex (d) at (0.8,0);
             \vertex (e) at (-0.8,0);
             \diagram* {
             (vq1) --   (vg1);
             (vq3) --   (vg2);
             (vq2) --   (vg1);
             (vq4) --   (vg2);
             (vq3) --   (vq2);
             (vq4) --   (vq1);
             (vq1) -- [dashed, dash pattern=on 1.5pt off 1.5pt] (v2);
             (v1) -- [dashed, dash pattern=on 1.5pt off 1.5pt] (vq3);
             (vq4) -- [dashed, dash pattern=on 1.5pt off 1.5pt] (vq2);
             (vg1) -- [dashed, dash pattern=on 1.5pt off 1.5pt] (e);
             (d) -- [dashed, dash pattern=on 1.5pt off 1.5pt] (vg2);
             };
    \end{feynman}
\end{tikzpicture}
    \end{minipage}
}

\newcommand{\nonppsi}{
\begin{minipage}{3.5cm}
\begin{tikzpicture}[scale=2]
    \begin{feynman}
             \vertex (vg1) at (-0.6,0);
             \vertex (vg2) at (0.6,0);
             \vertex (vq1) at (-0.3,0.3);
             \vertex (vq2) at (-0.3,-0.3);
             \vertex (vq3) at (0.3,-0.3);
             \vertex (vq4) at (0.3,0.3);
             \vertex (v1) at (0.05,-0.05);
             \vertex (v2) at (-0.05,0.05);
             \vertex (d) at (0.8,0);
             \vertex (e) at (-0.8,0);
             \diagram* {
             (vq1) --   (vg1);
             (vq3) --   (vg2);
             (vq2) -- [dashed, dash pattern=on 1.5pt off 1.5pt]  (vg1);
             (vq4) --  [dashed, dash pattern=on 1.5pt off 1.5pt] (vg2);
             (vq3) --   (vq2);
             (vq4) --   (vq1);
             (vq1) -- [dashed, dash pattern=on 1.5pt off 1.5pt] (v2);
             (v1) -- [dashed, dash pattern=on 1.5pt off 1.5pt] (vq3);
             (vq4) -- (vq2);
             (vg1) --  (e);
             (d) -- (vg2);
             };
    \end{feynman}
\end{tikzpicture}
    \end{minipage}
}

\newcommand{\nonvqua}{
\begin{tikzpicture}[scale = 2]
\draw (-0.5,0.1) node[scale=0.8]{$a$};
    \draw (0.5,0.1) node[scale=0.8]{$c$};
    \draw (0.05,-0.15) node[scale=0.8]{$b$};
    \draw (-0.2,0.05) node[scale=0.8]{$d$};
    \draw (0.2,0.05) node[scale=0.8]{$e$};
    \draw (0,0.38) node[scale=0.8]{$f$};
    \begin{feynman}
             \vertex (vg1) at (-0.4,-0.1);
             \vertex (vq1) at (0,-0.4);
             \vertex (vQ1) at (-0.4,0.3);
             \vertex (vg2) at (0.4,-0.1);
             \vertex (vq2) at (0,0);
             \vertex (vQ2) at (0.4,0.3);
             \vertex (vs) at (0,0.3);
             \vertex (d) at (0.85,-0.1);
             \vertex (e) at (-0.85,-0.1);
             \vertex (f) at (0,0.6);
             \diagram* {
             (e) -- [fermion] (vg1);
             (vg2) -- [fermion] (d);
             (vg1) -- [fermion] (vq1);
             (vq1) --  [fermion] (vg2);
             (vQ2) --  (vQ1);
             (vQ1) --  (vq2);
             (vq2) --  (vQ2);
             (vg1) -- [scalar] (vQ1);
            (vQ2) -- [scalar] (vg2);
             (vq1) -- [scalar] (vq2);
             };
    \end{feynman}
\end{tikzpicture}
}

\newcommand{\flescapeo}{\begin{minipage}{5cm}
\begin{tikzpicture}[scale = 3]
        \begin{feynman}
             \vertex (vg1) at (-0.45,0);
             \vertex (vq1) at (-0.2,0.25);
             \vertex (vQ1) at (-0.2,-0.25);
             \vertex (vQ3) at (-0.2,0);
             \vertex (vq3) at (0.2,0);
             \vertex (vg2) at (0.45,0);
             \vertex (vq2) at (0.2,0.25);
             \vertex (vQ2) at (0.2,-0.25);
             \vertex (d) at (0.65,0);
             \vertex (e) at (-0.65,0);
             \diagram* {
             (vg1) --[scalar] (e);
             (vg1) -- [thick] (vq1);
             (vq1) -- [thick] (vQ1);
             (vQ1) -- [thick] (vg1);
             (vq1) -- [scalar] (vq2);
             (vQ1) -- [scalar] (vQ2);
             (vg2) -- [thick] (vq2);
             (vq2) -- [thick] (vQ2);
             (vq3) -- [scalar] (vQ3);
             (vQ2) -- [thick] (vg2);
             (vg2) --[scalar] (d);
             };
          \end{feynman}
\end{tikzpicture}
\end{minipage}}

\newcommand{\flescapet}{\begin{minipage}{5cm}
\begin{tikzpicture}[scale=3]
        \begin{feynman}
             \vertex (vg1) at (-0.45,0);
             \vertex (vq1) at (-0.2,0.25);
             \vertex (vQ1) at (-0.2,-0.25);
             \vertex (vQ3) at (-0.2,0);
             \vertex (vq3) at (0.2,0);
             \vertex (vg2) at (0.45,0);
             \vertex (vq2) at (0.2,0.25);
             \vertex (vQ2) at (0.2,-0.25);
             \vertex (d) at (0.65,0);
             \vertex (e) at (-0.65,0);
             \diagram* {
             (e) -- [fermion] (vg1);
             (vq1) -- [fermion] (vg1);
             (vq1) -- [scalar] (vQ3);
             (vQ1) -- [scalar] (vg1);
             (vq1) -- [fermion] (vq2);
             (vQ1) -- [thick] (vQ2);
             (vg2) -- [fermion] (vq2);
             (vq2) -- [scalar] (vq3);
             (vQ3) -- [thick] (vQ1);
             (vq3) -- [thick] (vQ2);
             (vq3) -- [thick] (vQ3);
             (vQ2) -- [scalar] (vg2);
             (vg2) -- [fermion] (d);
             };
          \end{feynman}
\end{tikzpicture}
\end{minipage}}

%% file: main.bbl
\begin{thebibliography}{10}

\bibitem{Martin:1997ns}
S.P.~Martin, \emph{{A Supersymmetry primer}},
  \href{https://doi.org/10.1142/9789812839657_0001}{\emph{Adv. Ser. Direct.
  High Energy Phys.} {\bfseries 18} (1998) 1}
  [\href{https://arxiv.org/abs/hep-ph/9709356}{{\ttfamily hep-ph/9709356}}].

\bibitem{ParticleDataGroup:2024cfk}
{\scshape Particle Data Group} collaboration, \emph{{Review of particle
  physics}}, \href{https://doi.org/10.1103/PhysRevD.110.030001}{\emph{Phys.
  Rev. D} {\bfseries 110} (2024) 030001}.

\bibitem{Gross:1974jv}
D.J.~Gross and A.~Neveu, \emph{{Dynamical Symmetry Breaking in Asymptotically
  Free Field Theories}},
  \href{https://doi.org/10.1103/PhysRevD.10.3235}{\emph{Phys. Rev. D}
  {\bfseries 10} (1974) 3235}.

\bibitem{Zerf:2017zqi}
N.~Zerf, L.N.~Mihaila, P.~Marquard, I.F.~Herbut and M.M.~Scherer,
  \emph{{Four-loop critical exponents for the Gross-Neveu-Yukawa models}},
  \href{https://doi.org/10.1103/PhysRevD.96.096010}{\emph{Phys. Rev. D}
  {\bfseries 96} (2017) 096010}
  [\href{https://arxiv.org/abs/1709.05057}{{\ttfamily 1709.05057}}].

\bibitem{Gracey:2025aoj}
J.A.~Gracey, A.~Maier, P.~Marquard and Y.~Schr{\"o}der, \emph{{Anomalous
  dimensions and critical exponents for the Gross-Neveu-Yukawa model at five
  loops}}, \href{https://doi.org/10.1103/lmpx-n3mj}{\emph{Phys. Rev. D}
  {\bfseries 112} (2025) 085029}
  [\href{https://arxiv.org/abs/2507.22594}{{\ttfamily 2507.22594}}].

\bibitem{Lee:2006if}
S.-S.~Lee, \emph{{Emergence of supersymmetry at a critical point of a lattice
  model}}, \href{https://doi.org/10.1103/PhysRevB.76.075103}{\emph{Phys. Rev.
  B} {\bfseries 76} (2007) 075103}
  [\href{https://arxiv.org/abs/cond-mat/0611658}{{\ttfamily
  cond-mat/0611658}}].

\bibitem{Grover:2013rc}
T.~Grover, D.N.~Sheng and A.~Vishwanath, \emph{{Emergent Space-Time
  Supersymmetry at the Boundary of a Topological Phase}},
  \href{https://doi.org/10.1126/science.1248253}{\emph{Science} {\bfseries 344}
  (2014) 280} [\href{https://arxiv.org/abs/1301.7449}{{\ttfamily 1301.7449}}].

\bibitem{Zerf:2016fti}
N.~Zerf, C.-H.~Lin and J.~Maciejko, \emph{{Superconducting quantum criticality
  of topological surface states at three loops}},
  \href{https://doi.org/10.1103/PhysRevB.94.205106}{\emph{Phys. Rev. B}
  {\bfseries 94} (2016) 205106}
  [\href{https://arxiv.org/abs/1605.09423}{{\ttfamily 1605.09423}}].

\bibitem{Grisaru:1979wc}
M.T.~Grisaru, W.~Siegel and M.~Rocek, \emph{{Improved Methods for
  Supergraphs}},
  \href{https://doi.org/10.1016/0550-3213(79)90344-4}{\emph{Nucl. Phys. B}
  {\bfseries 159} (1979) 429}.

\bibitem{Poole:2019kcm}
C.~Poole and A.E.~Thomsen, \emph{{Constraints on 3- and 4-loop
  $\beta$-functions in a general four-dimensional Quantum Field Theory}},
  \href{https://doi.org/10.1007/JHEP09(2019)055}{\emph{JHEP} {\bfseries 09}
  (2019) 055} [\href{https://arxiv.org/abs/1906.04625}{{\ttfamily
  1906.04625}}].

\bibitem{tHooft:1972tcz}
G.~'t~Hooft and M.J.G.~Veltman, \emph{{Regularization and Renormalization of
  Gauge Fields}},
  \href{https://doi.org/10.1016/0550-3213(72)90279-9}{\emph{Nucl. Phys. B}
  {\bfseries 44} (1972) 189}.

\bibitem{Bollini:1972ui}
C.G.~Bollini and J.J.~Giambiagi, \emph{{Dimensional Renormalization: The Number
  of Dimensions as a Regularizing Parameter}},
  \href{https://doi.org/10.1007/BF02895558}{\emph{Nuovo Cim. B} {\bfseries 12}
  (1972) 20}.

\bibitem{Breitenlohner:1977hr}
P.~Breitenlohner and D.~Maison, \emph{{Dimensional Renormalization and the
  Action Principle}}, \href{https://doi.org/10.1007/BF01609069}{\emph{Commun.
  Math. Phys.} {\bfseries 52} (1977) 11}.

\bibitem{Korner:1991sx}
J.G.~K{\"o}rner, D.~Kreimer and K.~Schilcher, \emph{{A practicable $\gamma_5$
  scheme in dimensional regularization}},
  \href{https://doi.org/10.1007/BF01559471}{\emph{Z. Phys. C} {\bfseries 54}
  (1992) 503}.

\bibitem{Avdeev:1982jx}
L.V.~Avdeev, S.G.~Gorishnii, A.Y.~Kamenshchik and S.A.~Larin, \emph{{Four-loop
  $\beta$ function in the {Wess-Zumino} Model}},
  \href{https://doi.org/10.1016/0370-2693(82)90727-4}{\emph{Phys. Lett. B}
  {\bfseries 117} (1982) 321}.

\bibitem{n2symat3loops}
M.~Chakraborty and S.-O.~Moch, \emph{{Dimensional Reduction is Supersymmetric
  at Three Loops}},  \href{https://arxiv.org/abs/2603.02892}{{\ttfamily
  2603.02892}}.

\bibitem{Passarino:1978jh}
G.~Passarino and M.J.G.~Veltman, \emph{{One Loop Corrections for $e^+ e^-$
  Annihilation Into $\mu^+ \mu^-$ in the Weinberg Model}},
  \href{https://doi.org/10.1016/0550-3213(79)90234-7}{\emph{Nucl. Phys. B}
  {\bfseries 160} (1979) 151}.

\bibitem{Kamenshchik:1983ufs}
A.Y.~Kamenshchik, \emph{{Dimensional regularization of chiral superfields}},
  \href{https://doi.org/10.1007/BF01015801}{\emph{Theor. Math. Phys.}
  {\bfseries 55} (1983) 431}.

\bibitem{Avdeev:1982xy}
L.V.~Avdeev and A.A.~Vladimirov, \emph{{Dimensional Regularization and
  Supersymmetry}},
  \href{https://doi.org/10.1016/0550-3213(83)90437-6}{\emph{Nucl. Phys. B}
  {\bfseries 219} (1983) 262}.

\bibitem{Velizhanin:2008rw}
V.N.~Velizhanin, \emph{{Three-loop renormalization of the N=1, N=2, N=4
  supersymmetric Yang-Mills theories}},
  \href{https://doi.org/10.1016/j.nuclphysb.2009.03.017}{\emph{Nucl. Phys. B}
  {\bfseries 818} (2009) 95} [\href{https://arxiv.org/abs/0809.2509}{{\ttfamily
  0809.2509}}].

\bibitem{Nogueira:1991ex}
P.~Nogueira, \emph{{Automatic Feynman Graph Generation}},
  \href{https://doi.org/10.1006/jcph.1993.1074}{\emph{J. Comput. Phys.}
  {\bfseries 105} (1993) 279}.

\bibitem{Vermaseren:2000nd}
J.A.M.~Vermaseren, \emph{{New features of FORM}},
  \href{https://arxiv.org/abs/math-ph/0010025}{{\ttfamily math-ph/0010025}}.

\bibitem{Kuipers:2012rf}
J.~Kuipers, T.~Ueda, J.A.M.~Vermaseren and J.~Vollinga, \emph{{\texttt{FORM}
  version 4.0}}, \href{https://doi.org/10.1016/j.cpc.2012.12.028}{\emph{Comput.
  Phys. Commun.} {\bfseries 184} (2013) 1453}
  [\href{https://arxiv.org/abs/1203.6543}{{\ttfamily 1203.6543}}].

\bibitem{Ruijl:2017dtg}
B.~Ruijl, T.~Ueda and J.~Vermaseren, \emph{{\texttt{FORM} version 4.2}},
  \href{https://arxiv.org/abs/1707.06453}{{\ttfamily 1707.06453}}.

\bibitem{Davies:2026cci}
J.~Davies, T.~Kaneko, C.~Marinissen, T.~Ueda and J.A.M.~Vermaseren,
  \emph{{\texttt{FORM} Version 5.0}},
  \href{https://arxiv.org/abs/2601.19982}{{\ttfamily 2601.19982}}.

\bibitem{minos}
J.~Vermaseren, \emph{{ \texttt{Minos}, {\rm
  {\url{https://www.nikhef.nl/~form/maindir/others/minos/minos.html}}} }}, .

\bibitem{Ruijl:2017cxj}
B.~Ruijl, T.~Ueda and J.A.M.~Vermaseren, \emph{{Forcer, a FORM program for the
  parametric reduction of four-loop massless propagator diagrams}},
  \href{https://doi.org/10.1016/j.cpc.2020.107198}{\emph{Comput. Phys. Commun.}
  {\bfseries 253} (2020) 107198}
  [\href{https://arxiv.org/abs/1704.06650}{{\ttfamily 1704.06650}}].

\bibitem{Mihaila:2017ble}
L.N.~Mihaila, N.~Zerf, B.~Ihrig, I.F.~Herbut and M.M.~Scherer,
  \emph{{Gross-Neveu-Yukawa model at three loops and Ising critical behavior of
  Dirac systems}},
  \href{https://doi.org/10.1103/PhysRevB.96.165133}{\emph{Phys. Rev. B}
  {\bfseries 96} (2017) 165133}
  [\href{https://arxiv.org/abs/1703.08801}{{\ttfamily 1703.08801}}].

\bibitem{Baikov:2010hf}
P.A.~Baikov and K.G.~Chetyrkin, \emph{{Four Loop Massless Propagators: An
  Algebraic Evaluation of All Master Integrals}},
  \href{https://doi.org/10.1016/j.nuclphysb.2010.05.004}{\emph{Nucl. Phys. B}
  {\bfseries 837} (2010) 186}
  [\href{https://arxiv.org/abs/1004.1153}{{\ttfamily 1004.1153}}].

\bibitem{Falcioni:2025hfz}
G.~Falcioni, F.~Herzog, S.~Moch, A.~Pelloni and A.~Vogt, \emph{{Additional
  results on the four-loop flavour-singlet splitting functions in QCD}},
  \href{https://arxiv.org/abs/2512.10783}{{\ttfamily 2512.10783}}.

\bibitem{Moch:2018wjh}
S.~Moch, B.~Ruijl, T.~Ueda, J.A.M.~Vermaseren and A.~Vogt, \emph{{On quartic
  colour factors in splitting functions and the gluon cusp anomalous
  dimension}},
  \href{https://doi.org/10.1016/j.physletb.2018.06.017}{\emph{Phys. Lett. B}
  {\bfseries 782} (2018) 627}
  [\href{https://arxiv.org/abs/1805.09638}{{\ttfamily 1805.09638}}].

\bibitem{Onishchenko:2003hs}
A.I.~Onishchenko and V.N.~Velizhanin, \emph{{Anomalous dimensions of twist-2
  conformal operators in supersymmetric Wess-Zumino model}},
  \href{https://arxiv.org/abs/hep-ph/0309222}{{\ttfamily hep-ph/0309222}}.

\bibitem{Vogelsang:1996im}
W.~Vogelsang, \emph{{The Spin dependent two loop splitting functions}},
  \href{https://doi.org/10.1016/0550-3213(96)00306-9}{\emph{Nucl. Phys. B}
  {\bfseries 475} (1996) 47}
  [\href{https://arxiv.org/abs/hep-ph/9603366}{{\ttfamily hep-ph/9603366}}].

\bibitem{Manashov:2025kgf}
A.N.~Manashov, S.~Moch and L.A.~Shumilov, \emph{{Anomalous dimensions at small
  spins}}, \href{https://doi.org/10.1007/JHEP09(2025)106}{\emph{JHEP}
  {\bfseries 09} (2025) 106}
  [\href{https://arxiv.org/abs/2506.05132}{{\ttfamily 2506.05132}}].

\bibitem{Vermaseren:1998uu}
J.A.M.~Vermaseren, \emph{{Harmonic sums, Mellin transforms and integrals}},
  \href{https://doi.org/10.1142/S0217751X99001032}{\emph{Int. J. Mod. Phys. A}
  {\bfseries 14} (1999) 2037}
  [\href{https://arxiv.org/abs/hep-ph/9806280}{{\ttfamily hep-ph/9806280}}].

\bibitem{Velizhanin:2010cm}
V.N.~Velizhanin, \emph{{Six-Loop Anomalous Dimension of Twist-Three Operators
  in N=4 SYM}}, \href{https://doi.org/10.1007/JHEP11(2010)129}{\emph{JHEP}
  {\bfseries 11} (2010) 129} [\href{https://arxiv.org/abs/1003.4717}{{\ttfamily
  1003.4717}}].

\bibitem{Siegel:1979wq}
W.~Siegel, \emph{{Supersymmetric Dimensional Regularization via Dimensional
  Reduction}}, \href{https://doi.org/10.1016/0370-2693(79)90282-X}{\emph{Phys.
  Lett. B} {\bfseries 84} (1979) 193}.

\bibitem{Siegel:1980qs}
W.~Siegel, \emph{{Inconsistency of Supersymmetric Dimensional Regularization}},
  \href{https://doi.org/10.1016/0370-2693(80)90819-9}{\emph{Phys. Lett. B}
  {\bfseries 94} (1980) 37}.

\bibitem{Novikov:1983uc}
V.A.~Novikov, M.A.~Shifman, A.I.~Vainshtein and V.I.~Zakharov, \emph{{Exact
  Gell-Mann-Low Function of Supersymmetric Yang-Mills Theories from Instanton
  Calculus}}, \href{https://doi.org/10.1016/0550-3213(83)90338-3}{\emph{Nucl.
  Phys. B} {\bfseries 229} (1983) 381}.

\end{thebibliography}
